\documentclass[lettersize,journal]{IEEEtran}
\usepackage{amsmath,amsfonts}
\usepackage{algorithmic}
\usepackage{algorithm}
\usepackage{array}
\usepackage{textcomp}
\usepackage{stfloats}
\usepackage{url}
\usepackage{verbatim}
\usepackage{graphicx}
\usepackage{cite}
\usepackage{graphicx} 
\graphicspath{{img/}} 
\usepackage{mathtools}
\usepackage[caption=false,font=footnotesize]{subfig}

\allowdisplaybreaks
\usepackage{textcomp}
\usepackage{microtype}
\usepackage{amsthm}
\usepackage{amssymb}
\newtheorem{theorem}{Theorem}

\usepackage{booktabs,makecell,tabularx}

\newcolumntype{C}[1]{>{\centering\arraybackslash}p{#1}}
\newcolumntype{L}{>{\raggedright\arraybackslash}X}
\newcommand{\khnote}[1]{{\bf\color{red}[ #1 -- Kaiyang ]}}
\usepackage{siunitx}
\usepackage{etoolbox}
\newrobustcmd{\B}{\bfseries}
\allowdisplaybreaks
\usepackage[dvipsnames]{xcolor}

\usepackage{hyperref}
\usepackage{amsthm}
\theoremstyle{remark}
\newtheorem{remark}{Remark}
\usepackage{stfloats}
\usepackage{textcomp}
\usepackage{balance}

\begin{document}

\title{A Heterogeneous Multiscale Method for Efficient Simulation of Power Systems with Inverter-Based Resources}

\author{Kaiyang Huang,~\IEEEmembership{Student Member,~IEEE, }Min Xiong,~\IEEEmembership{Member,~IEEE, }Yang Liu,~\IEEEmembership{Member,~IEEE, }

Kai Sun,~\IEEEmembership{Fellow,~IEEE. }
\thanks{This work was supported in part by NSF grant ECCS-2329924 and in part by the ISSE seed grant of the University of Tennessee, Knoxville.\textit{(Corresponding author: Kai Sun.)}}
\thanks{K. Huang, M. Xiong, Y. Liu and K. Sun are with the Department of Electrical Engineering and Computer
Science, University of Tennessee, Knoxville, TN 37996 USA (e-mail: khuang12@vols.utk.edu; mxiong3@vols.utk.edu; yang.powersystems@gmail.com; kaisun@utk.edu).}}

\markboth{Accepted by IEEE Transactions on Power Systems (DOI: 10.1109/TPWRS.2025.3539567)}%
{Shell \MakeLowercase{\textit{et al.}}: A Sample Article Using IEEEtran.cls for IEEE Journals}


\maketitle

\begin{abstract}
As inverter-based resources (IBRs) penetrate power systems, the dynamics become more complex, exhibiting multiple timescales, including electromagnetic transient (EMT) dynamics of power electronic controllers and electromechanical dynamics of synchronous generators. Consequently, the power system model becomes highly stiff, posing a challenge for efficient simulation using existing methods that focus on dynamics within a single timescale. This paper proposes a Heterogeneous Multiscale Method for highly efficient multi-timescale simulation of a power system represented by its EMT model. The new method alternates between the microscopic EMT model of the system and an automatically reduced macroscopic model, varying the step size accordingly to achieve significant acceleration while maintaining accuracy in both fast and slow dynamics of interests. It also incorporates a semi-analytical solution method to enable a more adaptive variable-step mechanism. The new simulation method is illustrated using a two-area system and is then tested on a detailed EMT model of the IEEE 39-bus system.

\end{abstract}


\section{Introduction}
\IEEEPARstart{I}{nverter}-based resources (IBRs), including renewables and energy storage systems, are rapidly
increasing their presence in power grids worldwide for low-carbon, clean electricity generation. This trend is challenging the traditional dominance of conventional power plants. However, as IBRs become more prevalent in power grids, the analysis and control of power grid dynamics become more complex due to their increased timescales, as well as the high dimensionality and nonlinearity of power grid models. During contingencies, power grid dynamics can span a wide spectrum, from electromagnetic transient (EMT) dynamics with power electronic controllers of IBRs in microseconds to quasi-steady-state grid behaviors in minutes or hours to balance demand and supply against uncertainties of renewables. These dynamics are respectively 1,000-10,000 times faster and slower than transient stability related electromechanical dynamics of synchronous generators\cite{Hatziargyriou2021}. Time-domain simulation is the most important tool for assessing the stability and reliability of a power grid against threats from multiple timescales. However, simulating a real-world power grid with high
    IBR penetration poses significant challenges to the accuracy and efficiency requirements. If its dynamics of all timescales, such as EMT dynamics and electro-mechanical oscillations, are taken into account, this will result in a large number of nonlinear ordinary differential equations (ODEs) with a high stiffness ratio. Thus, a step size in microseconds is often required for numerical stability. Additionally, simulating a series of quasi-steady-state behaviors of the grid mixed with inverter controls and inter-area oscillations for an extended period of tens of minutes will result in an extremely slow simulation due to a huge computational burden.
Therefore, the best practice of multi-timescale simulations today for power system engineers is still co-simulation, which defines interfaces between EMT models, transient stability models, and quasi-steady-state models and co-simulates them iteratively using different simulation tools, with each tool focusing on a portion of the grid in one timescale while reducing the rest\cite{Qiuhua2016,Qiuhua2018}. However, it is still an open question how to determine the interface boundaries. Furthermore, the overall convergence, accuracy, and time performance for such co-simulations still pose major technical challenges\cite{Marandi2009}.
\par 
Fully detailed simulations such as EMT simulations of an interconnected power grid as a single, nonlinear stiff system to capture its dynamics in all timescales are highly complex tasks\cite{Lin2009}. In practice, power system engineers typically focus the modeling and simulating a utility-scale grid on a specific timescale for, e.g., EMT controls, transient stability and quasi-steady-state analysis, depending on the purpose of the study.
When EMT simulations are required for the integration studies of IBRs, a reduced EMT model can be offline developed, which maintains details only for a specific region of concern, while the rest of the grid is either equivalenced or ignored beyond its boundary \cite{Marandi2009,Lin2009,Yi2013}. 
Simulating of a power system using a manually reduced model pose accuracy concerns. Conventionally, models for longer timescale simulations \cite{Kai2006} are generated manually from circuit-based EMT models by developers. The structures and parameters of these models are pending identification and validation using field event data, which, if inaccurate, can lead to incorrect simulation results. For power systems with high-penetration of IBRs, model reduction becomes even more challenging due to two factors. First, a phasor-based postive-sequence model is mainly used in transient stability simulations becomes less accurate, even for some not very fast dynamics, particularly when IBRs induce high-frequency oscillations in the sub-synchronous frequency range (e.g. 5-60 Hz). These oscillations simulated on a phasor-based model can manifest a significant sacrifice in accuracy \cite{Hatziargyriou2021}. Second, some IBR models provided by the developers are black box models encoded to protect intellectual properties, which are hard to reduce using a traditional model reduction method for mathematical equations.

For efficient and accurate simulations of a power system with IBRs, this paper proposes a multi-timescale simulation approach built upon the Heterogeneous Multiscale Method (HMM) \cite{Engquist2005HeterogeneousMM}, which provides a general framework for designing multiscale algorithms to model and simulate complex, stiff nonlinear dynamical systems. Over the past two decades, the HMM has demonstrated its efficacy in simulating complex systems across various fields of science and engineering, particularly when macroscopic behaviors are of major interest but only microscopic models are available or accurate \cite{Engquist2005HeterogeneousMM,Ariel2008AMM}. 
Unlike \cite{Engquist2005HeterogeneousMM}, this paper for the first time proposes a novel HMM framework along with its theoretical foundation and a scalable simulation approach for efficient two-timescale simulations of a power system represented in its EMT model. The proposed method enables an automatic, case-specific reduction of a detailed microscopic EMT model into a macroscopic model during the simulation, significantly improving computational efficiency while preserving essential dynamics. The macro-state in the proposed HMM framework is linked to the original micro-state via two transformations with a flexible dimension, which is lower than or equal to that of the micro-state. Additionally, the framework introduces a kernel-based convolution method to approximate macro-dynamics from micro-dynamics without requiring an explicitly reduced macro-model, enabling smooth and robust transitions between the micro-process and macro-process. Furthermore, the paper establishes the theoretical foundation of the method, uncovering relationships between simulation accuracy, computational complexity, and speedup through rigorous theorems. The framework also introduces variable-step algorithms to simulate both micro-state and macro-state efficiently. By leveraging a semi-analytical solution (SAS) method, a new theorem ensures the accurate capture of critical EMT dynamics while dynamically skipping unimportant microscopic dynamics for maximum acceleration in micro-state simulations.

\par The rest of the paper is organized as follows: Section II briefly introduces a generic HMM framework and then provides a detailed presentation of the proposed HMM approach for efficient two-timescale simulations of power systems represented in EMT models. The section includes the main idea and steps, the overall flowchart of the approach, the power system model considered in this paper, detailed algorithms for the micro-process and macro-process of the approach on the power system model, and a discussion of the main differences from other multi-timescale methods for power system simulations. Section III validates the proposed HMM approach on three test systems represented in their EMT models: a two-area system, the IEEE 39-bus system, and a 390-bus large system, demonstrating the performance and scalability of the proposed approach. Finally, Section IV presents the conclusions.

\color{black}
\section{Heterogenous Multiscale Method for Simulating Power Systems in EMT Models}
\subsection{Introduction of a generic HMM}
Consider a stiff ODE system of two or more timescales:
\begin{equation}
\label{stiff_ode}
    \dot{ \mathbf{x}}=\mathbf{f}\left(\mathbf{x}, t\right)
    \Longleftrightarrow \left\{\begin{array}{l}\dot{\mathbf{x}}^I=\mathbf{f}^I(\mathbf{x}^{I}, \mathbf{x}^{II},t) \\ \boldsymbol{\epsilon} \dot{\mathbf{x}}^{II}=\mathbf{f}^{I I}(\mathbf{x}^{I}, \mathbf{x}^{II},t)\end{array}\right.\\,
   \vspace{-2mm}
\end{equation}
where $\mathbf{x}(t): \mathbb{R}^{+} \mapsto \mathbb{R}^{n_x}$ is the state variable vector, and $n_x$ is the dimension of the system. The system has dynamics of a fast timescale so that the system can be represented as a two-timescale system with a slow vector $\mathbf{x}^I$ and a fast vector $\mathbf{x}^{II}$, which are distinguished by a nonsingular matrix $\boldsymbol{\epsilon}$ having eigenvalues close to zero \cite{Engquist2005HeterogeneousMM}.
The stiffness of a dynamical system described in the form of \eqref{stiff_ode} can be defined based on the eigenvalues of the Jacobian of $\mathbf{f}$. Specifically, stiffness is measured as the ratio of the largest eigenvalue to the smallest eigenvalue in terms of absolution value, i.e. $\left|\lambda_{\max } /\left|/ \lambda_{\min }\right|\right.$. For a two-timescale system like \eqref{stiff_ode}, the eigenvalues separate into two groups due to the time-scale gap defined by $\boldsymbol{\epsilon}$. Hence, the stiffness ratio is approximately $O\left(\|\boldsymbol{\epsilon}^{-1}\|\right) \gg 1$, in particular $\|\boldsymbol{\epsilon}^{-1}\|$ represents the largest absolute value of eigenvalues of $\boldsymbol{\epsilon}^{-1}$ when 2-norm $\left(\|\cdot\|_{2}\right)$ is applied.

This paper focuses on accelerating power system simulations that include both fast EMT dynamics arising from, e.g., IBRs and the network and relatively slow electromechanical dynamics associated with synchronous generators. These dynamics span time scales from approximately $10^{-6} \mathrm{~s}$ to $10^2 \mathrm{~s}$, resulting in a stiffness ratio approaching $10^8$ or higher. To address the stiffness of power system models with IBRs for efficient simulation, two types of fast dynamics need to be considered: 1) dissipative EMT transients characterized by large eigenvalues with negative real parts, which occur immediately after a disturbance like a fault or switching, and 2) sustained oscillatory dynamics characterized by eigenvalues having near-zero real parts, whose examples include instantaneous current and voltage waveforms around the synchronous frequency and subsynchronous oscillations with IBRs. Conventional EMT simulation tools face challenges in efficiently simulating such two-timescale power system models over extended time periods due to their stiffness. These tools often require very small step sizes, in the microsecond range, to maintain numerical stability while simulating continuously until slow electromechanical dynamics, in the range of tens of seconds to minutes, conclude and the system approaches its steady-state condition. In contrast, the proposed HMM approach leverages the time-scale gap defined by $\boldsymbol{\epsilon}$ and introduces a macro-process to enable more efficient simulations once dissipative EMT transients subside or oscillatory dynamics reach a steady state.  
\color{black}  Suppose the existence of an effective macro-model as a reduction of \eqref{stiff_ode} with a vector $\mathbf{u}: \mathbb{R}^{+} \mapsto \mathbb{R}^{n_u}$ focusing on slow dynamics when $\|\boldsymbol{\epsilon}\|$ converges to $0$
\vspace{-1mm}
\begin{equation}
\label{macro}
    \dot{\mathbf{u}}=\bar{\mathbf{f}}(\mathbf{u}, t).
\vspace{-1mm}
\end{equation}
To distinguish it from the original system\eqref{stiff_ode}, the state of the original system \eqref{stiff_ode} is defined as the micro-state, and directly solving the original micro-model is called the micro-process. \color{black}An engineering approach for model reduction may assume fast dynamics are all stable and ignore their state variables in \eqref{stiff_ode} by simply letting $\boldsymbol{\epsilon}=\mathbf{0}$. Thus, the second equation in \eqref{stiff_ode} becomes an algebraic equation, reducing the stiff ODE system to a DAE (differential-algebraic equation) system, which is the conventional form of a power system model for transient stability simulations. 

\color{black}
Comparatively, the HMM approach focuses on finding the authentic $\mathbf{u}$ and its vector field $\bar{\mathbf{f}}$ that characterize a macro-model as manifested from the real-time response of the stiff system. For instance, in \cite{Engquist2005HeterogeneousMM}, it is assumed that $\mathbf{u} \equiv \mathbf{x}$ while the effective force $\bar{\mathbf{f}}$ on slow dynamics is obtained by:
\begin{equation}
\label{def_force_1}
\bar{\mathbf{f}}(t)=\lim _{\delta \rightarrow 0}\left(\lim _{\|\boldsymbol{\epsilon}\| \rightarrow 0} \frac{1}{\delta} \int_t^{t+\delta} \mathbf{f}\left(\mathbf{x}, \tau\right) d \tau\right).
\end{equation}
However, it is often time-consuming to derive $\bar{\mathbf{f}}$ by \eqref{def_force_1}. Instead, the HMM suggests evaluating it ``on the fly" of simulating \eqref{stiff_ode}. Namely, the analytical form of $\bar{\mathbf{f}}(t)$, i.e., an explicit reduction of \eqref{stiff_ode} is neither used nor required. Rather, $\bar{\mathbf{f}}(t)$ is estimated along the simulation as an implicit reduction, and hence it is adaptive to each simulated case regarding averaged macroscopic dynamics with $\mathbf{u}$ while important microscopic dynamics are still regained by solving \eqref{stiff_ode} and \eqref{macro} interactively. 
For instance, Ref. \cite{Engquist2005HeterogeneousMM} processes $\mathbf{f}_{}$ using a kernel function ${K}^{p, q} \in \mathbb{{K}}^{p, q}(I)$ to relax all dissipative modes and average stiff dynamics over a sliding time interval for estimating $\bar{\mathbf{f}}(t)$. ${K^{p, q}}$ is assumed to be $q^\text{th}$-order differentiable (i.e., $\in C_c^q(\mathbb{R})$) satisfying \eqref{intgral_condition}, with a compact support: $\operatorname{supp}({K^{p, q}})=I$, i.e. $K^{p, q}(t)\equiv0$ for $t\in \mathbb{R}\backslash I$. 
\begin{equation}
\label{intgral_condition}
\int_{\mathbb{R}} K^{p, q}(t) t^r d t= \begin{cases}1, & r=0 \\ 0, & 1 \leq r \leq p .\end{cases}
\end{equation}
The kernel can be scaled by $\eta$ with a similar shape as:
\begin{equation}
K^{p, q}_\eta(t):=\frac{2}{\eta} K^{p, q}\left(\frac{2t}{\eta}\right),
\end{equation}
then, with $\eta \rightarrow 0 \text { as } |\boldsymbol{\epsilon}| \rightarrow 0$, $\bar{\mathbf{f}}$ is estimated by a convolution between the kernel and $\mathbf{f}_{}$:
\begin{equation}
\label{approximation}
K_\eta^{p, q} * \mathbf{f}=K_\eta^{p, q} *\left(\bar{\mathbf{f}}+\Tilde{\mathbf{f}}_{\epsilon}\right) \rightarrow \bar{\mathbf{f}} \text { as } \|\boldsymbol{\epsilon}\| \rightarrow 0,
\end{equation}
where $\Tilde{\mathbf{f}}_{\epsilon}(t)$ represents fast dynamics in \eqref{stiff_ode} that need to be averaged, including dissipative transient and oscillatory dynamics. 

A schematic HMM algorithm solves \eqref{stiff_ode} by repeatedly taking these two steps is illustrated by Fig. \ref{fig:A 2-timescale HMM scheme} and the following:
\par
\noindent \textbf{Step 1}: Evolve the micro-state to estimate macro-force: (i) Reconstruct the initial micro-state $\mathbf{x}^I_0=\mathbf{x}^I\left(t_n\right), \mathbf{x}^{II}_0=\mathbf{x}^{II}\left(t_n\right)$; (ii) Evolve the micro-state for an interval $\eta$ by solving \eqref{stiff_ode} for $t \in\left[t_n, t_n+\eta\right]$ with $\mathbf{x}^I_0, \mathbf{x}^{II}_0$ using a numerical solver (i.e. the ``Micro-Solver" in Fig. \ref{fig:A 2-timescale HMM scheme}) at a time step of $h$; (iii) Average the vector field of \eqref{stiff_ode} on the fly to generate $\overline{\mathbf{f}}$ via this convolution:
$$
K_{\eta}^{p, q} *\left[\begin{array}{c}
\mathbf{f}^I \\
\boldsymbol{\epsilon}^{-1} \mathbf{f}^{I I}
\end{array}\right]\left(\mathbf{x}^I\left({\Delta}\right), \mathbf{x}^{II}\left({\Delta}\right), {\Delta}\right) \rightarrow \overline{\mathbf{f}}\left(t^{\prime}_{n}\right),
$$
where $\Delta = t_n+\delta t \in \left[t_n, t_n+\eta\right]$ represents the evaluation point of the convolution operation.
\par
\noindent \textbf{Step 2}: To evolve the macro-state, computing $\mathbf{x}^{I}_{n+1}$ and $\mathbf{x}^{II}_{n+1}$ at $t=t_{n+1}$ using past macro-states $\mathbf{x}^{I}_n, \mathbf{x}^{II}_{n}$ and their corresponding $\overline{\mathbf{f}}$ by a numerical solver (i.e. the ``Macro-Solver") at a time step $H>>h$, and then let $n=n+1$. 
\begin{figure}
    \centering
    \includegraphics[width=0.8\linewidth]{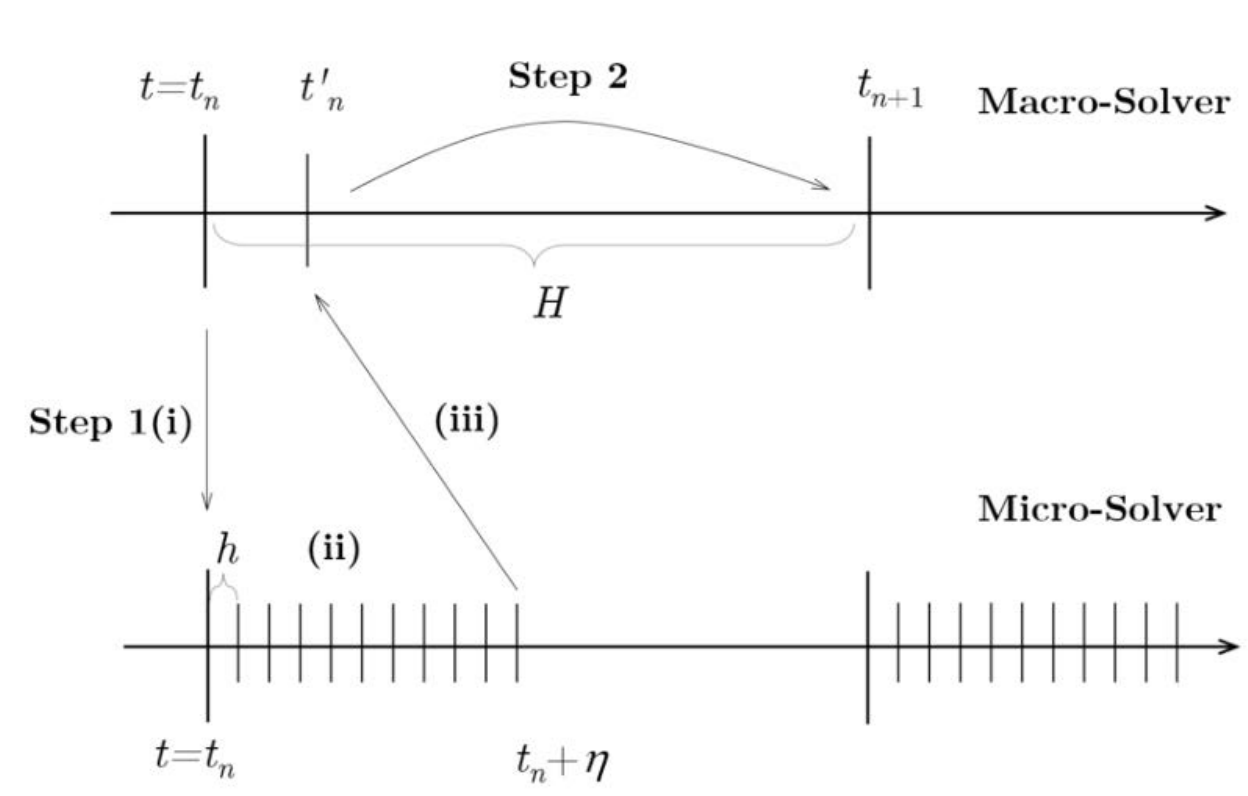}
    \caption{A two-timescale HMM scheme.}
    \label{fig:A 2-timescale HMM scheme}
    \vspace{-3mm}
\end{figure}

\color{black}  
For an EMT model in the form of \eqref{stiff_ode}, its eigenvalues exhibit two distinct time scales, corresponding to the slow electromechanical dynamics associated with the state variables in $\mathrm{x}^I$, and the fast EMT dynamics associated with the state variables in $\mathrm{x}^{I I}$. In general, fast state variables in $\mathrm{x}^{I I}$ include instantaneous currents of inductive components, instantaneous voltages of capacitive components and the terminal variables controlled by power electronic converters. Slow state variables in $\mathrm{x}^I$ are typically associated with synchronous generators and their controllers, which have much slower time constants compared to $\mathrm{x}^{I I}$. For the optimal performance of the proposed HMM approach, the separation between $x^I$ and $x^{I I}$ can be fine-tuned. In particular, in the proposed new HMM framework, the macro-state and micro-state share the same $\mathbf{x}^I$ (i.e., slow state variables), unlike existing HMM algorithms.
\color{black}

In the rest of this section, subsection II-\textit{B} will provide an overview of the proposed HMM approach for efficiently simulating EMT models, which includes its main idea and motivations to improve a generic HMM approach, basic steps, and its theoretical speedup on EMT simulations and computational complexity. Then, II-\textit{C} introduces the EMT model considered in the paper for a power system with IBRs. II-\textit{D} and II-\textit{E} respectively present the detailed micro-process and macro-process algorithms of the proposed HMM approach applied to the EMT model. 

\subsection{Overview of the Proposed HMM Approach}

\vspace{1mm}
\textit{1) Main Idea and Motivations}

\color{black}
Over the past two decades, the generic HMM concept, as discussed in Section II-A, has been successfully applied in various complex systems that exhibit obvious multi-timescale dynamics, including complex fluids \cite{REN20051}, dynamics of solids at finite temperature \cite{Weinan07}, and the three-body problem \cite{Engquist2005HeterogeneousMM}. However, as outlined in Section II-B, the original HMM framework proposed in the field of applied mathematics cannot be directly applied to power systems as highly stiff nonlinear systems networking a large variety of resources and devices. For the first time, this paper proposes a scalable HMM approach for simulations of power systems in EMT models.
\color{black}The proposed HMM approach for simulating the EMT model of power systems will build upon the generic HMM framework, with sufficient considerations given to EMT models and simulation requirements. Introduce an intermediate variable denoted by $\mathbf{u}_\epsilon$, which are linked to the mciro-state via a compression transformation $Q$ and a reconstruction transformation $R$:
\vspace{-3mm}
\begin{equation}
\begin{aligned}
    \label{RandQ}
    \mathbf{u}_\epsilon = Q(\mathbf{x}),     \mathbf{x} = R(\mathbf{u}_\epsilon),
\end{aligned}
\vspace{-1mm}
\end{equation}
$Q$ and $R$ can be inverses of each other when $n_u = n_x$. Note that $\mathbf{u}_\epsilon$ can still involve fast dynamics unless $\boldsymbol{\epsilon}=0$, so its vector field $\hat{\mathbf{f}}_{\epsilon}: \mathbb{R}^{n_{u}}\times\mathbb{R}^{+} \mapsto \mathbb{R}^{n_u}$ may still contain two timescales. Then, the macro-state $\mathbf{u}$ can be obtained based on the intermediate state and the macro-force:
\begin{equation}
\label{stiff_macro}
    \dot{\mathbf{u}}_\epsilon=\hat{\mathbf{f}}_{\epsilon}\left(\mathbf{u}_\epsilon, t\right)\rightarrow \dot{\mathbf{u}}=\bar{\mathbf{f}}\left(\mathbf{u}, t\right) \text{ as }\|\boldsymbol{\epsilon}\| \text{ converges to 0.}
\end{equation}

Unlike \cite{Engquist2005HeterogeneousMM}, we assume that the macro-state may have a lower dimension than $\mathbf{x}_{}$, the micro-state of \eqref{stiff_ode}, i.e. $n_u \leq n_x$, 
Thus, the effective fore can be defined by 
\begin{equation}
\label{def_force}
\bar{\mathbf{f}}(t)=\lim _{\delta \rightarrow 0}\left(\lim _{\epsilon \rightarrow 0} \frac{1}{\delta} \int_t^{t+\delta} \hat{\mathbf{f}}_\epsilon\left(Q\mathbf{x}, \tau\right) d \tau\right),
\end{equation}
whose estimation, to avoid time-consuming computation, requires an algorithm similar to \eqref{approximation} but more sophisticated and considering power systems.

In bulk power system simulations assessing electromechanical transient stability, a DAE model that ignores the faster dynamics with $\mathbf{x}^{II}$ by letting $\boldsymbol{\epsilon}=0$ is used. Thus, the ODE on $\mathbf{x}^{II}$ is reduced to an algebraic equation acting as constraints and control laws of the grid such as the power-flow equation. Comparatively, in the HMM framework, the fast transient dynamics are averaged for instantaneous voltages and currents of all capacitive or inductive components of the network. Even though there is still $0 \approx \mathbf{f}^{I I}$ after averaging, it indicates a stable limit cycle for instantaneous voltages and currents varying at about the ac fundamental frequency 60 Hz. Further, by appropriately selecting transformations $Q$ and $R$, phasor-like macro-states on voltages and currents can be obtained as more meaningful and slower variables from three-phase time-varying waveforms of voltages and currents, similar to the envelopes of voltages and currents defined in transient stability simulations.

However, the generic HMM \cite{Engquist2005HeterogeneousMM} introduced above is unable to yield such macro-state variables since its $R$ and $Q$ are assumed to be identical operators with $\mathbf{u} \equiv \mathbf{x}$ as $|\boldsymbol{\epsilon}|$ converges to 0. Thus, $\hat{\mathbf{f}}_{\epsilon}\equiv \mathbf{f}$ in \eqref{approximation}, meaning the macro-force $\bar{\mathbf{f}}$ is approximated by solving the stiff dynamic system \eqref{stiff_ode} directly. This implies a further derivation of \eqref{stiff_macro} for EMT models. For instance, in our preliminary work \cite{Kaiyang}, the generic HMM is not directly applied to a conventional $abc$-frame EMT model; rather, a derived $0dq$-frame EMT model is used. In this paper, the proposed HMM approach is able to estimate macro-dynamics from computed micro-dynamics directly by selecting a kernel function properly. Thus, there is no need to know \eqref{stiff_macro} explicitly, consistent with the motivation of the HMM. 
 Approximate $\bar{\mathbf{f}}$ by a time series of $\mathbf{u}_\epsilon=Q(\mathbf{x})$. Considering the strong symmetry in EMT waveforms on the micro-state, a symmetric variant bump kernel function:
\vspace{-2mm}
\begin{equation}
\label{kernel}
 {K}(t)= \begin{cases}C e^{-\frac{D}{1-t^2}}, &-1 < t < 1 \\ 0, &\lvert t \rvert \geq 1.\end{cases} 
\vspace{-1mm}
\end{equation}
is used to average the microscopic dynamics over a sliding time window of $\eta$ to generate $\overline{\mathbf{f}}$, where $C$ and $D$ are adjustable parameters to satisfy \eqref{intgral_condition}. By utilizing the above kernel, a new way to calculate the derivative of slow variables can be provided. For an arbitrary time interval $\Omega=\left[t_n, t_n+\eta\right]$, let $\Delta = t_n+\delta t=t_n+\eta/2$, we have
\vspace{-2mm}
\begin{equation}
    \begin{aligned}
    \label{integral_slow_de}
K_\eta* \dot{\mathbf{u}}_\epsilon(\Delta)
&=\int_{\mathbb{R}} \dot{\mathbf{u}}_\epsilon(\tau) K_\eta(\Delta-\tau) d \tau\\
&=\int_{t_n}^{t_n+\eta} \dot{\mathbf{u}}_\epsilon(\tau) K_\eta(\Delta-\tau) d \tau\\
& =\left.K_\eta(\Delta-\tau) \mathbf{u}_\epsilon(\tau)\right|_{\Omega}-\int_{t_n}^{t_n+\eta} \mathbf{u}_\epsilon(\tau) d K_\eta(\Delta-\tau)\\
&=\int_{t_n}^{t_n+\eta}\mathbf{u}_\epsilon(\tau) {K}^{\prime}_\eta(\Delta-\tau) d \tau \\
& ={K}^{\prime}_\eta * \mathbf{u}_\epsilon(\Delta),
\end{aligned}
\end{equation}
where the first line is based on the definition of the convolution operation at $\Delta$; the second and fourth lines hold due to the compact $\operatorname{supp}(K_\eta)=\left[-\eta/2, \eta/2 \right]$; the third line follows from the Integration by Parts. Benefiting from \eqref{integral_slow_de}, the calculation of $K_\eta*\hat{\mathbf{f}}_{\epsilon}(\Delta)$ can be replaced by ${K}^{\prime}_\eta * \mathbf{u}_\epsilon(\Delta)$, which prevents the issue that the analytical or numerical form of $\hat{\mathbf{f}}_{\epsilon}$ in \eqref{stiff_macro} may be unknown when $R$ and $Q$ are not identical operators. 

\vspace{1mm}
\textit{2) Steps of the Proposed HMM Approach}

The proposed HMM approach for EMT models in the form of \eqref{stiff_ode} adopts a solution process repeatedly performing the same two steps in an alternative way as illustrated by Fig.~\ref{fig:A 2-timescale HMM scheme}. 
\par
\noindent\textbf{Step 1 (Micro-Process)}: Estimate the macro-force by: (i) reconstructing $\mathbf{x}_{0}=R\left(\mathbf{u}_{\epsilon, n}\right)$ at $t=t_n$ by a linear or nonlinear transformation $R$; (ii) evolving the micro-state for interval $\eta$ by solving \eqref{stiff_ode} for $t \in\left[t_n, t_n+\eta\right]$ with $\mathbf{x}\left(t_n\right)=\mathbf{x}_{0}$ using an $r$-th order numerical or semi-analytical micro-solver at a time step of $h$ for the EMT simulation, and (iii) constructing time series of $\mathbf{u}_{\epsilon}$ in $\left[t_n, t_n+\eta\right]$ by $\mathbf{u}_{\epsilon}=Q\left(\mathbf{x}\right)$ and averaging the vector field of \eqref{stiff_ode} on the fly to generate $\overline{\mathbf{f}}$ at $t_n^{\prime}=t_n+\eta$ via:
\begin{equation}
\label{kernel conv}
    {K}^{\prime}_\eta * \mathbf{u}_\epsilon(\Delta)=K_\eta* \hat{\mathbf{f}}_{\epsilon}\left(\mathbf{u}_\epsilon(\Delta), \Delta\right) \rightarrow \overline{\mathbf{f}}\left(t^{\prime}_{n}\right),
\end{equation}
Here $\Delta = t_n+\eta/2$ unlike the generic HMM.
\par
\noindent\textbf{Step 2 (Macro-Process)}: Evolve the macro-state by computing $\mathbf{u}_{{n+1}}$ at $t=t_{n+1}$ by (2) using past macro-states $\mathbf{u}_{n}, \mathbf{u}_{n-1}, \ldots, \mathbf{u}_{n-k}$ (or their estimates by the micro-solver at a later instance such as $t_n^{\prime}$ for $\mathbf{u}_{\epsilon}$ ) and their corresponding $\overline{\mathbf{f}}$ by an $S$-th order macro-solver at a time step $H>>h$, and then let $n=n+1$ until termination time $T$ is approached. \color{black}The overall steps are also summarized and shown in the flowchart Fig.~\ref{fig:flowchart}.\color{black}

\color{black}
\begin{remark}

The proposed HMM approach ensures smooth two-way transitions between the micro-process simulation (i.e., the full EMT simulation) and the macro-process simulation focusing on slower dynamics through a compression transformation $Q$ that compresses the micro-state to a macro-state and its inverse, a reconstruction transformation $R$ back to the micro-state. Detailed EMT dynamics are resolved in the micro-process using a small micro-step $h$ for a duration of $\eta$, and the results are averaged via a kernel function \eqref{kernel conv} to estimate the macro-force driving the macro-state to be simulated over a much larger, macro-step $H$. Smoothness and numerical stability with the transitions can be ensured by the selection of a kernel function based on, e.g., \cite[$\bold{Lemma~2.2}$]{Engquist2005HeterogeneousMM}.
\end{remark}
\color{black}

\vspace{1mm}
\begin{figure}
    \centering
    \includegraphics[width=1\linewidth]{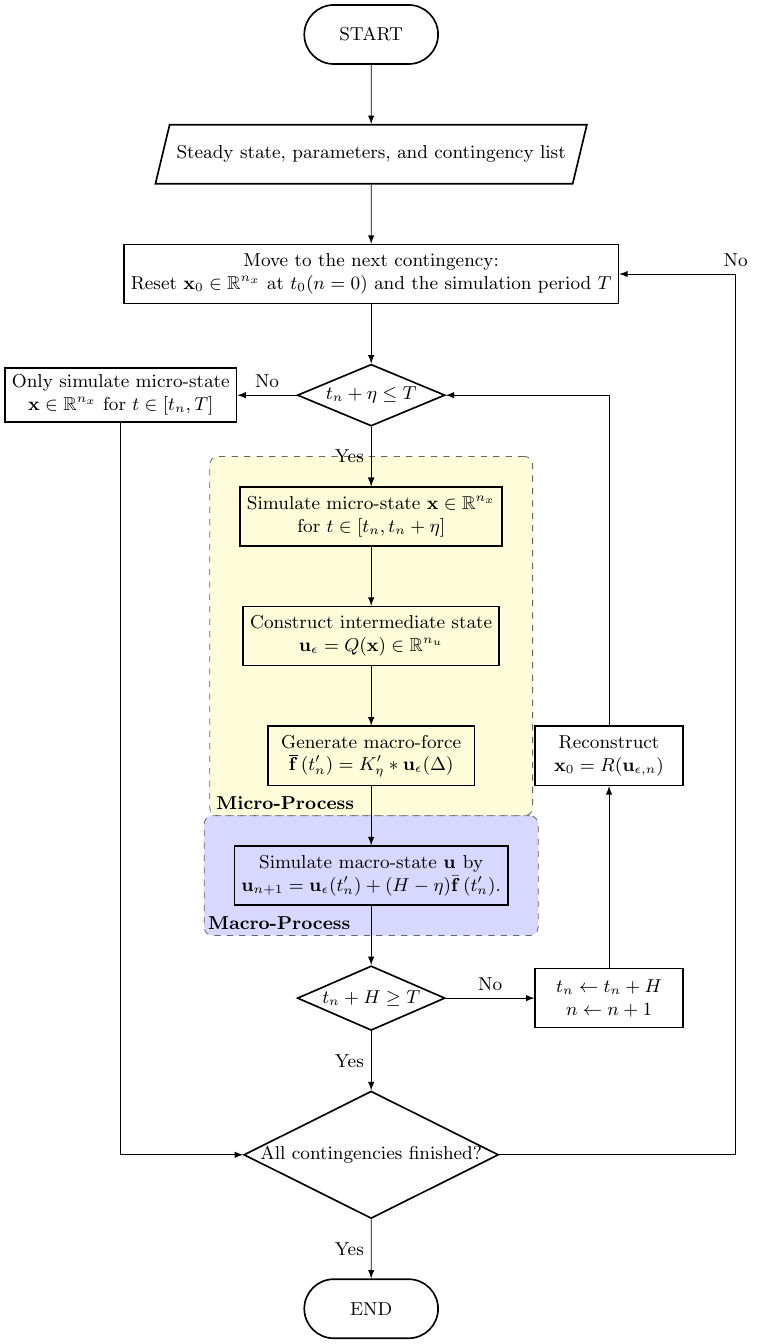}
    \color{black}
    \caption{Flowchart of the proposed HMM approach}
    \label{fig:flowchart}
    \vspace{-4mm}
\end{figure}
\textit{3) Speedup and Complexity}       

A detailed full EMT simulation using the same EMT model and the same time step as $h$ with the micro-solver is equivalent to the proposed HMM approach having the micro-process spans the entire simulation period. Comparatively, the proposed HMM approach can speed up the simulation by involving successive macro-processes using a much larger time step of $H$. Hence, the times of speedup will increase with the increase of $H$ and the decrease of the interval $\eta$. Besides these two, the computational complexity of the approach also depends on the micro-solver time step $h$. Define the kernel estimation error $|\mathcal{E}_{H M M}|= \| {K}^{\prime}_\eta * \mathbf{u}_{\epsilon}(\Delta)- \overline{\mathbf{f}}\left(t^{\prime}_{n}\right)\|$, for a well-defined HMM approach, it should be bounded by $H^{S}$, so that it is comparable to the local truncation error of the given macro-scheme. 
There is a following theorem on the theoretical speedup and computational complexity of the simulation using the proposed HMM approach compared to a full EMT simulation. 
\begin{theorem}
Suppose the micro-process uses an $r$-th order solver with a time step of $h$, the macro-process uses an $S$-th order solver with a time step of $H$, and a $k$-th order differentiable kernel is applied to the first $\eta$ interval of each micro-process, the HMM approach---alternating between the micro-process and macro-process via $Q$ and $R$, both being invertible linear transformations---has these estimates on its computational complexity and speedup compared to a solution process involving only the micro-process:
\begin{equation}
\label{Complexity}
\vspace{-2mm}
\text { Complexity }=\frac{\eta}{H h} \sim O\left(\|\boldsymbol{\epsilon}\|^{-\frac{r+k+1}{k r}} H^{-\frac{r+k+1}{k r} S-1}\right) 
\end{equation}
\begin{equation}
\label{speedup}
\text {Speedup }=\frac{H}{\eta} \sim O\left(\|\boldsymbol{\epsilon}\|^{\frac{1}{k}-1} H^{\frac{S}{k}+1}\right)
\end{equation}
such that the kernel estimation error $|\mathcal{E}_{\color{black}\text{HMM}}| \leq C_K H^S$, where $C_K$ is a constant.
\end{theorem}
\begin{proof}
Suppose $K \in \mathbb{K}_\eta^{p, k}$ is used. 
 For a given micro-solver, the error in approximating $\mathbf{x}_{}$ in Step 1(ii) is $C_\text{micro} h^r \eta \|\boldsymbol{\epsilon}\|^{-(r+1)}$, $C_\text{micro}$ is a constant,
where the term $\|\boldsymbol{\epsilon}\|^{-(r+1)}$ comes from $d^{(r+1)}\mathbf{x}_{} / d t^{(r+1)}$, notice that $Q$ is bounded, so the error in approximating $\mathbf{u}_\epsilon$ is also bounded, namely, $\|  Q\mathbf{x}_{}(t_n)- Q\mathbf{x}_{n}\| \leq \| Q\|{ }_{\infty} C h^r \eta \|\boldsymbol{\epsilon}\|^{-(r+1)}$, thus we have the error $\mathcal{E}_{\text {micro }}$ in evaluating $\hat{\mathbf{f}}_{\epsilon}$ accumulated in Step 1 (ii) and (iii) satisfies
\begin{equation}
\begin{aligned}
\left|\mathcal{E}_{\text{micro}}\right| 
&\leq C_\text{micro} \| Q\|_{\infty} \left\|\frac{\partial \hat{\mathbf{f}}_{\epsilon}}{\partial \mathbf{u}_\epsilon}\right\|_{\infty} 
\eta h^r \|\boldsymbol{\epsilon}\|^{-(r+1)} \\
&\leq \tilde{C} \eta h^r \|\boldsymbol{\epsilon}\|^{-(r+2)} .
\end{aligned}
\end{equation}

Hence, it is required $\tilde{C}\eta h^r \|\boldsymbol{\epsilon}\|^{-(r+2)} \sim H^S$. Omitting the constant, this leads to:
\begin{equation}
\label{est1}
h\sim O\left(\eta^{-1 / r} H^{S/ r} \|\boldsymbol{\epsilon}\|^{1+2 / r} \right)\text {. }
\end{equation}
From \cite[$\bold{Theorem~2.7}$]{Engquist2005HeterogeneousMM}, the force estimation error using a kernel in $\mathbb{K}^{p, k}$ is
\begin{equation}
\mathcal{E}_{\color{black}\text{HMM}}^{(a)} \leq C_1 \eta^p+C_2\left(\frac{\|\boldsymbol{\epsilon}\|}{\eta}\right)^k,    
\end{equation}
where \(C_2\) is proportional to \(\left\|\partial \hat{\mathbf{f}}_{\epsilon} / \partial \mathbf{u}_{\epsilon}\right\|_{\infty}\), then we have
$C_2\left(\|\boldsymbol{\epsilon}\|/\eta\right)^k = \tilde{C}_2 \|\boldsymbol{\epsilon}\|^{k-1}/\eta^k.$
Ignoring the constants, observe that the term \(\|\boldsymbol{\epsilon}\|^{k-1} / \eta^k\) becomes dominant over \(\eta^p\) for \(\eta < \eta^*(\|\boldsymbol{\epsilon}\|)\), allowing for the minimization of the number of micro-processes. In this case, the dominating term is
\begin{equation}
\label{est2}
    \frac{\|\boldsymbol{\epsilon}\|^{k-1}}{\eta^k} \sim H^S \Longrightarrow \eta \sim H^{-S / k} \|\boldsymbol{\epsilon}\|^{(k-1) / k} .
\end{equation}
Notice that given $H$ is fixed, it is easy to know that \eqref{est1} and \eqref{est2} hold if and only if \eqref{Complexity} and \eqref{speedup} hold. Hence, with conditions \eqref{Complexity} and \eqref{speedup}, $|\mathcal{E}_{\text{\color{black}HMM}}|\leq \left|\mathcal{E}_{\text {micro}}\right|+ |\mathcal{E}_{\text{\color{black}HMM}}^{(a)}|\leq C_K H^S$. 
\end{proof}
\color{black}
\textbf{Theorem 1} can be applied to any simulation tool that implements the proposed HMM approach as long as the tool offers variable-order and variable-step solvers or allows for the adjustment of those parameters including $r$, $S$ and $H$.    
\color{black}
 \color{black}
\begin{remark}
     When $k \rightarrow \infty$, i.e., the kernel function becomes smooth, the complexity and speedup of the HMM framework simplify to:
\begin{equation}
\label{Complexity_reduced}
\vspace{-2mm}
\text { Complexity }=\frac{\eta}{H h} \sim O\left(\|\boldsymbol{\epsilon}\|^{-\frac{1}{r}} H^{-\frac{1}{r} S-1}\right), 
\end{equation}
\begin{equation}
\label{speedup_reduced}
\text {Speedup }=\frac{H}{\eta} \sim O\left(\|\boldsymbol{\epsilon}\|^{-1} H\right).
\end{equation}
When $r \gg S$, increasing $r$ reduces the overall complexity of the HMM approach because a higher order micro-solver enables more accurate simulation of the micro-process even with a larger micro-step $h$. Consequently, the number of micro-steps can be reduced, decreasing the total computational time. However, increasing the order might require more computation for each step by the micro-solver. Therefore, the overall speedup needs to balance the tradeoff between the reduction in computational complexity due to fewer steps and the increase of the time per step caused by the higher-order micro-solver. For a power system with two time scales having a substantial gap defined by $\boldsymbol{\epsilon}$, where detailed EMT dynamics occur only for very limited durations through the simulation period, the overall speedup achieved using the HMM approach primarily depends on the macro-process simulation, which takes significantly larger macro-steps of $H$.
 \end{remark}
 \color{black}
\subsection{Power System Models}
This section introduces the EMT models of power systems considered in this paper for detailed designs and implementation of the proposed HMM approach.

\vspace{1mm}
\textit{1) Synchronous Generator and Controllers}  

The volage-behind-reactance model \cite{VBRmodel}, which is used in EMT simulations for generators,  is considered here to apply the proposed HMM approach. Each generator is modeled by these ODEs:
\begin{subequations}
\vspace{-2mm}
    \begin{align}
&\dot{\delta}=\Delta \omega_{r}, \\
\Delta &\dot{\omega}_{r}=\frac{\omega_0}{2 H_g}\left(p_{m}-p_{e}-D_g \frac{\Delta \omega_{r}}{\omega_0}\right) ,\\
&\dot{\lambda}_{f d }=e_{f d }-\frac{r_{f d}}{L_{l f}}\left(\lambda_{f d  }-\lambda_{a d}\right), \\
&\dot{\lambda}_{1 d}=-\frac{r_{1 d }}{L_{1 d l }}\left(\lambda_{1 d }-\lambda_{a d}\right) ,\\
&\dot{\lambda}_{1 q }=-\frac{r_{1 q }}{L_{1 q l }}\left(\lambda_{1 q }-\lambda_{a q}\right), \\
&\dot{\lambda}_{2 q }=-\frac{r_{2 q }}{L_{2 q l }}\left(\lambda_{2 q }-\lambda_{a q }\right), \\
\label{current injection}
&\dot{i}_{a b c}   =-L_{a b c }^{\prime \prime -1}\left(v_{a b c }-P^{-1} v_{0 d q}^{\prime \prime}+R_{s } i_{a b c}+\dot{L}_{a b c}^{\prime \prime} i_{a b c }\right),    
    \end{align}
\end{subequations}
where $\omega_0$ is the nominal frequency of the system, $\delta, \Delta \omega_{r}=\omega_{r}-\omega_{0}, H_g, D_g, p_{m}$, and $p_{e}$ represent the rotor angle, rotor speed deviation, inertial constant, damping coefficient, mechanical power, and electrical power of the generator, respectively.  $\lambda_{f d}, \lambda_{1 d}, \lambda_{1 q}$, and $\lambda_{2 q}$ are flux linkages of the filed winding, $d$-axis damper winding, $q$-axis first damper winding, $q$-axis second damper winding, respectively; $r_{f d }, r_{1 d}, r_{1 q}$, and $r_{2 q}$ are resistances and $L_{f d l}, L_{1 d l}, L_{1 q l}$, and $L_{2 q l}$ are leakage inductances of these four windings; $e_{f d}$ is the field voltage. Furthermore, $i_{a b c}$ and $v_{a b c}$ are three-phase terminal current and voltage that interface with the grid; $R_{S}$ is the constant stator resistance matrix and $P(\theta)$ is the Park transformation defined by 
\begin{equation}
P(\theta)=\frac{2}{3}\left[\begin{array}{ccc}
1 / 2 & 1 / 2 & 1 / 2 \\
\cos (\theta) & \cos (\theta-2 \pi / 3) & \cos (\theta+2 \pi / 3) \\
-\sin (\theta) & -\sin (\theta-2 \pi / 3) & -\sin (\theta+2 \pi / 3)
\end{array}\right]
\end{equation}
with the reference angle $\theta$ satisfying $\dot{ \theta}=\omega_r$. Note that the subtransient inductance matrix $L_{a b c }^{\prime \prime}$ periodically changes with $\theta$. The related detailed calculation of $p_e$, $L_{a b c }^{\prime \prime}$ and subtransient voltage $ v_{0 d q}^{\prime \prime}$ are shown in \cite{VBRmodel}. Without loss of generality, the TGOV1 turbine-governor model \cite{governor} and SEXS exciter model \cite{sexs} are considered for generator controls in this paper.  
\color{black}

\textit{2) Inverter-Based Resource}

\begin{figure}
    \centering
    \includegraphics[width=0.8\linewidth]{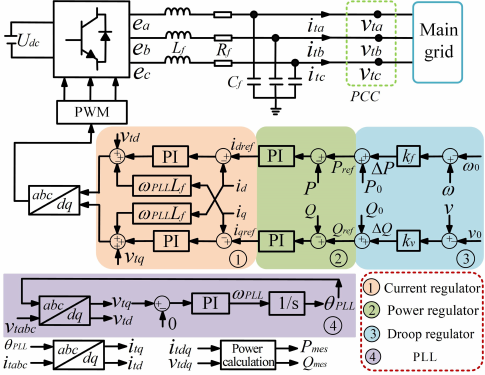}
    \caption{Diagram of a grid-following IBR model}
    \label{fig:IBR_digm}
    \vspace{-4mm}
\end{figure}
This paper considers the grid-following IBR model in Fig.~\ref{fig:IBR_digm}, which comprises an outer-loop current regulator, an inner-loop power regulator, a frequency droop controller, and a voltage droop controller\cite{Buxin}. The dynamics of its Pulse-width modulation (PWM) are disregarded in the system-level simulation and analysis. Since Fig.~\ref{fig:IBR_digm} can illustrate the control loop in
 the frequency domain clearly, the mathematical model in the
 state space are not presented in detail here for the notation
 simplicity. For instance, a Phase-Locked Loop (PLL) is used to track the bus voltage angle. Note that the subscript $ref$ represents the reference value, and $0$ signifies the steady-state value. The PLL control loop is described as 
\vspace{-3mm}
\begin{equation}
\begin{aligned}
&\dot{\theta}_{PLL}=\omega_{PLL},\\
&\dot{\phi} = v_{tq},\\
&\omega_{PLL}=\omega_{\mathrm{s}}+K_{P L L}^{\mathrm{p}}\dot{\phi}+K_{P L L}^{\mathrm{i}} \phi,
\end{aligned}
\vspace{-1.5mm}
\end{equation}
where $\theta_{PLL}$ is the angle of the point of common coupling (PCC) voltage, $\omega_{PLL}$ is the frequency of the inverter, $\phi$ is the intermediate variable in the PLL control loop, and $K_{P L L}^{\mathrm{P}}$ and $K_{P L L}^{\mathrm{I}}$ are the corresponding PI gains of the PLL loop, $v_{tq}$ is the $q$-axis PCC voltage under a local Park transformation $P(\theta_{PLL})$. For the current regulator, power regulator, and droop regulator, the detailed information used in this paper can be found in \cite{Buxin,Modelreduction}.

\vspace{1mm}
\color{black}
\textit{3) Network}  

To model the power network, the $\Pi$-model is considered for each line, which includes R-L circuits and shunt capacitances in $abc$-frame. Loads in EMT models are also regarded as R-L or R-C circuits. 

The whole network is modeled as a directed graph $\mathcal{G}=(\mathcal{N}, \mathcal{P})$ with $N$ nodes and $P$ edges with an arbitrary direction assigned to each edge $l \in \mathcal{P}$. The incidence matrix $B\in\mathbb{R}^{N \times E}$ can represent the connectivity of the network topology from a one-line diagram view:
\vspace{-3mm}
$$
B_{k, l}:=\left\{\begin{array}{ll}
+1 & , k=i \\
-1 & , k=j \\
0 & , \text {otherwise, }
\end{array} \forall l=(i, j) \in \mathcal{P} .\right.
$$
Note that for a three-phase balanced system, the connectivity is identical for different phases, and the incidence matrix $B$ can be generalized easily as an incidence matrix  $\mathbf{B} \in \mathbb{R}^{3N \times 3E}$ for the three-phase network by mapping each entry of $B$ to the corresponding $3\times3$ matrix (i.e., $+1 \mapsto +I_{3\times3}$, $-1 \mapsto -I_{3\times3}$ and $0 \mapsto \mathbb{O}_{3\times 3}$), leading this three-phase state-space network equation:
\begin{equation}
\label{Network}
    \left[\begin{array}{ll}
\mathbf{C} & \\
& \mathbf{L}
\end{array}\right]\left[\begin{array}{l}
\dot{\mathbf{v}}_{abc} \\
\dot{\mathbf{w}}_{abc}
\end{array}\right]=\left[\begin{array}{cc}
-\mathbf{G} & -\mathbf{B} \\
\mathbf{B}^{\top} & -\mathbf{R}
\end{array}\right]\left[\begin{array}{l}
\mathbf{v}_{abc} \\
\mathbf{w}_{abc}
\end{array}\right]+\left[\begin{array}{c}
\mathbf{i}_{abc}  \\
\mathbb{O}_{3E}
\end{array}\right],
\end{equation}
where $\mathbf{v}_{abc}\in \mathbb{R}^{3N}$ and $\mathbf{w}_{abc}\in \mathbb{R}^{3E}$ are three-phase node voltages and line currents associated with the capacitive and inductive storage elements inside the network, $\mathbf{i}_{abc}$ is the current injection from generators and IBRs to the network. $\mathbf{C}\in \mathbb{R}^{3N \times 3N}$,$\mathbf{G} \in \mathbb{R}^{3N\times 3N}$, $\mathbf{L}\in \mathbb{R}^{3E\times 3E}$, and $\mathbf{R} \in \mathbb{R}^{3E\times 3E}$ are three-phase parameter matrices consisting of shunt capacitance $C_i$ and conductance $G_i=R_{i}^{-1}, \forall i\in \mathcal{N}$ and line inductance $L_{i j}$ and resistance $R_{i j}$, $\forall (i,j)\in \mathcal{P}$. Equ.~\eqref{Network} is derived from Kirchhoff’s current law and Kirchhoff’s voltage law with algebraic graph theory\cite{graph}. Notice that the $\mathbf{i}_{abc}$ represents the interconnection between sources and the network, for instance, a generator's current injection to the connected bus is induced by \eqref{current injection}.

\color{black}

While a detailed power system model is presented here, the proposed approach also applies to models of other forms or time scales. A distinguishing feature of the HMM framework is that it requires only a single ODE model of the power system, namely the full EMT model. In other words, both the micro-process and macro-process are simulated using the same model, eliminating the need for model reduction of individual generators, control systems, or IBRs. The macro-state is computed directly from the micro-state and is driven by a macro-force derived from a kernel function. This characteristic ensures that the proposed HMM framework is potentially applicable to any dynamical system for which an ODE model is available.

\color{black}Based on these power system models, the detailed algorithms of the proposed HMM approach are presented as follows.

\subsection{Proposed Algorithms for the Micro-Process}
The DT (Differential Transformation) method\cite{liu2019power} is applied as the micro-solver to simulate the micro-state in the proposed HMM approach. 

\vspace{1mm}
\textit{1) Construction and Reconstruction}

\color{black}
Notice that state variables  $\mathbf{v}_{abc}$, $\mathbf{w}_{abc}$, and $\mathbf{i}_{abc}$ can be transformed into $0dq$-frame via a global Park transformation $P(\theta_{\text {global}})$ using a global reference angle $\theta_{\text {global}}$, by which the dynamics associated with the synchronous frequency (e.g. 60Hz) are almost eliminated while fast EMT dynamics and slow electro-mechanical dynamics with the network state variables become more obvious  \cite{Park} even under balanced cases. Thus, a kernel function is needed to average the simulated micro-state. 
The existence of two time scales, e.g. EMT dynamics vs. electromechanical dynamics, in the EMT model helps design the compression transformation $Q$ from the micro-state to macro-state and the reconstruction transformation $R=Q^{-1}$. In this paper, for $\mathbf{x} = [ \mathbf{x}^{I}, \mathbf{x}^{II} ]^\top$, $\mathbf{x}^{I}$ on slow dynamics includes the state variables of generators and IBRs except for their terminal currents $\mathbf{i}_{abc}$; on fast dynamics, $\mathbf{x}^{II}:=[\mathbf{v}_{abc},~\mathbf{w}_{abc},~\mathbf{i}_{abc}]^{\top}$ includes all network state variables. Thus, the compression transformation $Q$ from the micro-state to macro-state can comprise a Park transformation for network variables and an identify matrix for the other variables. The compression and reconstruction transformations are defined accordingly as $Q\coloneqq \text{diag}\left(\mathbb{I},\mathbf{P}(\mathbf{\theta}_{\text {global}})\right)$ and $R\coloneqq \text{diag}\left(\mathbb{I}, \mathbf{P}^{-1}(\mathbf{\theta}_{\text {global}})\right)$ where $\mathbb{I}$ is the identical matrix with the same dimension as $\mathbf{x}^{I}$ and $\mathbf{P}(\mathbf{\theta}_{\text {global}})$ is the vectorized invertible Park transformation for all network dynamics such that
\begin{equation}
\left[\begin{array}{l}
\mathbf{v}_{0dq} \\
\mathbf{w}_{0dq}\\
\mathbf{i}_{0dq}
\end{array}\right]=\mathbf{P}(\mathbf{\theta}_{\text {global}})\left[\begin{array}{l}
\mathbf{v}_{abc} \\
\mathbf{w}_{abc}\\
\mathbf{i}_{abc}
\end{array}\right].
\end{equation}
We have the intermediate state $\mathbf{u}_\epsilon=Q(\mathbf{x})=\left[\mathbf{u}_\epsilon^I, \mathbf{u}_\epsilon^{I I}\right]^{\top}$ as a transformation from the micro-state where $\mathbf{u}_\epsilon^I=\mathbf{x}^I$ and $\mathbf{u}_\epsilon^{I I}:=\left[\mathbf{v}_{0 d q}, \mathbf{w}_{0 d q}, \mathbf{i}_{0 d q}\right]^{\top}$. Starting from this intermediate state, the macro-state $\mathbf{u}$ focusing on slow dynamics can evolve governed by a macro-force obtained by \eqref{kernel conv}. Note that in general, the angle $\theta_{\text {global}}$ is not necessarily identical for all different state variables. For instance, each interface variable $\mathbf{i}_{abc}$ can utilize the reference angle of the local $0 d q$ frame of the connected generator or IBR.
\color{black}
\color{black}

\vspace{1mm}
\textit{2) Simulating Micro-State}
 
 Consider the previous dynamic system \eqref{stiff_ode} which is Lipschitz on $\mathbb{R}^{n_x}$:
\vspace{-2mm}
\begin{equation}\label{or}
\begin{aligned}
    \dot{\mathbf{x}}=\mathbf{f}(\mathbf{x},t).
\end{aligned}
\vspace{-2mm}
\end{equation}
The micro process involves solving \eqref{or} in $\left[t_n, t_n+\eta\right]$. Given an initial state $\mathbf{x}_0$ at $t_n$ by reconstruction $\mathbf{x}_0=R\mathbf{u}_{\epsilon,n}$, the solution of \eqref{or} is characterized by:
\vspace{-1mm}
\begin{equation}
\begin{aligned}
\mathbf{x}\left(t\right)=\mathbf{\Phi}{(t,\bf{x_0})},
\end{aligned}
\vspace{-2mm}
\end{equation}
which satisfies the corresponding initial value condition:
\vspace{-1mm}
\begin{equation}
\begin{aligned}
& \mathbf{x}\left(t_n\right)=\mathbf{\Phi}{(t_n,\bf{x_0})}=\bf{x_0}.
\end{aligned}
\vspace{-1mm}
\end{equation}
Suppose that $\mathbf{\Phi}{(t,\bf{x_0})}$ is sufficiently smooth in a neighborhood $\mathcal{B}_\delta(t_{n},\bf{x_0})$, its power series can be expanded up to a certain order ${L}$ within $\mathcal{B}_\delta(t_{n},\bf{x_0})$:
\vspace{-2mm}
\begin{equation}
\label{tay}
\begin{split}
\mathbf{\Phi}(t,\mathbf{x_0}) &=\sum_{k=0}^{{L}}{\frac{1}{k!}\frac{\partial^{k}{\mathbf{\Phi}(t,\mathbf{x_0})}}{\partial{t}^{k}}\bigg|_{t=t_{n}}(t-t_{n})^k}\\
    &+\mathbf{R}(\mathbf{\Theta)}(t-t_{n})^{{L}+1},
\end{split}
\vspace{-1mm}
\end{equation}
where $\mathbf{R}(\mathbf{\Theta)}$ is a constant related to order ${L}$, $t_{n}$, and $t$. The approximated solution $\tilde{x}$ can be rewritten in a compact form:
\vspace{-2mm}
\begin{equation}
\label{tay2}
\vspace{-2mm}
\begin{aligned}
&\mathbf{x}(t) \approx  \tilde{\mathbf{x}}(\mathbf{x}_0,L,h):=\sum_{k=0}^{{L}}{\mathbf{X}[k]h^k} ,  \\
&\mathbf{X}[k]={\frac{1}{k!}\frac{\partial^{k}{{\mathbf{\Phi}}(t,\mathbf{x_0})}}{\partial{t}^{k}}\bigg|_{t=t_{n}}},  
\end{aligned}
\end{equation}
where $h=t-t_n$ represents the step size. The main task of the DT algorithm is to find a recursive way to solve $\mathbf{X}[k]$ up to a desirable order ${L}$ in a certain time interval around $t_{n}$. The basic steps of the DT algorithm are shown as follows:

\vspace{2mm}
\noindent{\textbf{i) Derivation of DT Formulation.}} Given an ODE system \eqref{or}, the corresponding DT formulation can be derived utilizing the characterization of the original ODE system and linear independence of polynomials' coefficients. Some fundamental rules to transform time domain functions into DT forms are summarized by \cite{liu2019power,liu2019solving}. For instance, consider vectorized functions $\boldsymbol{\sigma}(t)$, $\boldsymbol{\xi}(t)$, $\boldsymbol{\gamma}(t)$, $\mathbf{g}_{s }(t)$, $\mathbf{g}_{c }(t)$ with vector coefficients $\boldsymbol{\Sigma}[k]$, $\boldsymbol{\Gamma}[k]$, $\mathbf{G}_{c }[k]$,  $\mathbf{G}_{s }[k]$, mappings of the $k$-th order coefficients of compositional functions are:
\vspace{-2mm}
\begin{equation*}
\begin{aligned}
    &\text{ i) } \boldsymbol{\gamma}(t) =\boldsymbol{\sigma}(t) \pm \boldsymbol{\xi}(t) \leftrightarrow \boldsymbol{\Gamma}[k]=\boldsymbol{\Sigma}[k] \pm \boldsymbol{\Xi}[k];\\ 
    &\text{ii) } \boldsymbol{\gamma}(t) =\boldsymbol{\sigma}^\top \leftrightarrow \boldsymbol{\Gamma}[k]=\boldsymbol{\Sigma}[k]^\top;\\
    &\text { iii) } \mathbf{g}_{c}(t)=\cos (\boldsymbol{\sigma}(t))\leftrightarrow \mathbf{G}_{c}[k]= \sum_{j=1}^{k} \frac{j}{k}\boldsymbol{\Gamma}[j] \odot \mathbf{G}_{s}[k-j];\\
    &\text { iv) } \mathbf{g}_{s }(t)=\sin (\boldsymbol{\sigma}(t)) \leftrightarrow \mathbf{G}_{s }[k]=\sum_{i=1}^k \frac{j}{k}\boldsymbol{\Gamma}[j] \odot \mathbf{G}_{c }[k-j] \text {; }
\end{aligned}
\vspace{-2mm}
\end{equation*}
where $\odot$ represents the element-wise product of vectors.
With these DT rules and more shown in \cite{liu2019power}, the dynamic system $\dot{\mathbf{x}}=\mathbf{f}(\mathbf{x},t)$ can be transformed into DT-form $\mathbf{F}[k]$ derived in \cite{liu2019solving}:
\vspace{-2mm}
\begin{equation}
\begin{aligned}
\label{im_DT}
(k+1) \mathbf{X}[k+1] & =\mathbf{F}[k] \coloneqq \mathbf{F}(\mathbf{X}[j]), \quad j=0, \cdots, k ,
\end{aligned}
\end{equation}
which allows us to solve the DT formulation instead of the original system. Note that this is a general way to solve $\mathbf{x}(t)$ up to an arbitrary desirable order approximation. Then, coefficients $\mathbf{X}[k]$ are solved based on the explicit DT form \eqref{im_DT} recursively from $k=1$ to ${L}$, providing an ${L}^{\text{th}}$ order approximation solution of $\mathbf{\Phi}$ in $\mathcal{B}_\delta(t_{n},\bf{x_0})$.

\vspace{2mm}
\noindent{\textbf{ii) Variable-step micro-process simulation.}} This paper focuses on variable-step micro-process simulation since fixed-step algorithms have been widely adopted. After calculating power coefficients $\mathbf{X}[k]$, a forward result at $t_{n}+h$ is obtained by evaluating power series \eqref{tay2}, thanks to the desirable convergence and accuracy of power series approximations with a relative high ${L}$, the predefined step size $h$ is relatively large compared with traditional numerical methods. Consider solving \eqref{or} requires multiple time steps, the step size $h_m$ is defined as ${m}^\text{th}$ step at $n^\text{th}$ interval, so the result at ${m+1}^\text{th}$ step at $n^\text{th}$ interval is:
\vspace{-2mm}
\begin{equation}
\label{truncated series}
\begin{aligned}
\mathbf{x}_{m+1} = \tilde{\mathbf{x}}(\mathbf{x}_m,L,h_m).\\
\end{aligned}
\vspace{-2mm}
\end{equation}
Note that $\mathbf{x}_{m} = \mathbf{X}[0]$ is at $t_n+\sum_{j=1}^{m-1}h_j$ with $m\geq 1$, so suppose we are at $t_n+\sum_{j=1}^{m-1}h_j$ and the initial value $\mathbf{x}_{m}$ is known, then \eqref{truncated series} can be directly calculated. However, to maintain the accuracy towards uncertain dynamics, it is helpful to control time step $h_{m}$ within a certain range, so the difference error is bounded\cite{ENRIGHT2000159}:
\vspace{-2mm}
\begin{equation}
\label{defect control}
\begin{aligned}
    &\lVert \dot{\mathbf{x}}_{m+1}-\mathbf{f}(\mathbf{x}_{m+1},t_n+\sum_{j=1}^{m}h_j) \rVert_{\infty}&<\epsilon_1 ,\\
\end{aligned}
\vspace{-1mm}
\end{equation}
where $\epsilon_1$ is a predefined threshold of the error, $\dot{\mathbf{x}}_{m+1}$ can be calculated by \eqref{truncated series}. While \eqref{defect control} provides a general way to design control strategies, it is not necessary to control all equations when the dynamics of certain state variables are known much faster than others, especially in multi-scale problems. In this paper, the defect control of time steps of \eqref{Network} is proposed. The motivation is that the dynamics of internal state variables of generators and IBRs are slower, so the step size is highly limited by \eqref{Network}. By utilizing the DT algorithm, \eqref{defect control} has an explicit form provided by the following theorem. For the notation simplicity, define $\boldsymbol{\psi}:=\left[\mathbf{v}_{abc},~~\mathbf{w}_{abc}\right]^{\top}$ and $\boldsymbol{\lambda}:= \left[\mathbf{i}_{abc},~~
\mathbb{O}_{3E}\right]^{\top}$, also, the corresponding $k^\text{th}$ order coefficients are $\boldsymbol{\Psi}[k]$ and $\boldsymbol{\Lambda}[k]$, namely, $\boldsymbol{\psi}_{m+1}= \tilde{\boldsymbol{\psi}}(\boldsymbol{\psi}_m,L,h_m)$ and $\boldsymbol{\lambda}_{m+1}=\tilde{\boldsymbol{\lambda}}(\boldsymbol{\lambda}_m,L,h_m)$
and
\begin{equation}
   \mathbf{A}_{eq} :=\left[\begin{array}{ll}
\mathbf{C} & \\
& \mathbf{L}
\end{array}\right]^{-1}\left[\begin{array}{cc}
-\mathbf{G} & -\mathbf{B} \\
\mathbf{B}^{\top} & -\mathbf{R}
\end{array}\right],~
\mathbf{B}_{eq} :=\left[\begin{array}{ll}
\mathbf{C} & \\
& \mathbf{L}
\end{array}\right]^{-1}.
\end{equation}
Based on \eqref{defect control}, define the error induced by \eqref{truncated series} in \eqref{Network} as
\begin{equation}
    \mathcal{E}_{m}:= \lVert \dot{\boldsymbol{\psi}}_m-\mathbf{A}_{eq} \boldsymbol{\psi}_m-\mathbf{B}_{eq} \boldsymbol{\lambda}_m\|_{\infty}.
\end{equation}
The following theorem is provided.
\begin{theorem}
    Given the linear form of \eqref{Network} and the $L^\text{th}$ order DT algorithm, the error in \eqref{truncated series} can be estimated on the fly of the simulation once the highest coefficients are calculated, i.e.,
\begin{equation}
\begin{aligned}
\label{error_est}
\mathcal{E}_{m+1}&=
\lVert \dot{\boldsymbol{\psi}}_{m+1}-\mathbf{A}_{eq} \boldsymbol{\psi}_{m+1}-\mathbf{B}_{eq} \boldsymbol{\lambda}_{m+1}\|_{\infty}\\
&=\|\mathbf{A}_{eq} \boldsymbol{\Psi}[L]+\mathbf{B}_{eq} \boldsymbol{\Lambda}[L]\|_{\infty}(h_{m})^{L}.   
\end{aligned}
\end{equation}
\end{theorem}
\begin{proof}
    The corresponding DT formulation of \eqref{Network} is
\begin{equation}
\label{DTnet}
    \boldsymbol{\Psi}[k+1]=\mathbf{A}_{eq} \boldsymbol{\Psi}[k]+\mathbf{B}_{eq} \boldsymbol{\Lambda}[k], \forall~0\leq k\leq L-1.
\end{equation}
Notice 
\begin{equation}
\begin{aligned}
&\boldsymbol{\psi}_{m+1} =\sum_{k=0}^{L} \boldsymbol{\Psi}[k]\left(h_{m}\right)^k,\\
&\boldsymbol{\lambda}_{m+1} =\sum_{k=0}^{L} \boldsymbol{\Lambda}[k]\left(h_{m}\right)^k.
\end{aligned}
\end{equation}
Substituting them into \eqref{Network}, we have 
\begin{equation}
\begin{split}
    \mathcal{E}_{m+1}=&  \lVert \dot{\boldsymbol{\psi}}_m-\mathbf{A}_{eq} \boldsymbol{\psi}_m-\mathbf{B}_{eq} \boldsymbol{\lambda}_m\|_{\infty}\\
            =& \lVert\sum_{k=0}^{L-1} \boldsymbol{\Psi}[k+1]\left(h_{m}\right)^k\\
            &-\sum_{k=0}^{L}\left[ \mathbf{A}_{eq} \boldsymbol{\Psi}[k]\left(h_{m}\right)^k-\mathbf{B}_{eq}  \boldsymbol{\Lambda}[k]\left(h_{m}\right)^k\right]\|_{\infty}\\
            =&\lVert\{\sum_{k=0}^{L-1}\left( \boldsymbol{\Psi}[k+1]-\mathbf{A}_{eq} \boldsymbol{\Psi}[k]-\mathbf{B}_{eq}  \boldsymbol{\Lambda}[k]\right)\\
            &+\mathbf{A}_{eq} \boldsymbol{\Psi}[L]+\mathbf{B}_{eq} \boldsymbol{\Lambda}[L] \}\left\|_{\infty}(h_{m}\right)^k.
\end{split}
\end{equation}
With \eqref{DTnet}, the first sum term is zero and we have           
$\mathcal{E}_{m+1}=\|\mathbf{A}_{eq} \boldsymbol{\Psi}[L]+\mathbf{B}_{eq} \boldsymbol{\Lambda}[L]\|_{\infty}(h_{m})^{L}.$
\end{proof}
\color{black}
\textbf{Theorem 2} ensures the accuracy and numerical stability of the micro-process in the HMM approach when applied to any dynamical system modeled by ODEs, regardless of the simulation platform.     
\color{black}
Note that different from other numerical methods, \eqref{error_est} provides an explicit way to compute an estimated error, then the step size $h_{m}$ with the satisfactory error can be obtained by solving $\mathcal{E}_{m+1}=\epsilon_1$. 

\vspace{2mm}
\noindent{\textbf{iii) Evaluation at a New Point.}} After calculating the forward result, i.e., $\mathbf{x}_{m+1}$ at $t_n+\sum_{j=1}^{m}h_j$, new 0$^\text{th}$ coefficients $\mathbf{X}[0]$ at $t_n+\sum_{j=1}^{m}h_j$ are obtained. Repeat the previous steps, $\mathbf{x}_{m+2}$ can then be computed until the terminated time $t_n+\eta$ for the $n^\text{th}$ micro process is achieved. 

\textit{3) Kernel Averaging}

After time series of micro-state $\mathbf{x}(t)$ are simulated, the macro-state variables can be computed by $\mathbf{u}_{\epsilon}=Q(\mathbf{x})$. The proposed kernel averaging method provides a general way to approximate slow dynamics behind the fast dynamics for all state variables, but only a certain group of state variables are necessary to be averaged to save the computational time and complexity.
Recall state variables are separated into two groups: $\mathbf{x}=\left[\mathbf{x}^{I},~\mathbf{x}^{II} \right]^{\top}$ where $\mathbf{x}^{II}=[\mathbf{v}_{abc},~\mathbf{w}_{abc},~\mathbf{i}_{abc}]^\top$. Here vector fields of two groups are defined as $\mathbf{f}^{I}$ and $\mathbf{f}^{II}$ respectively. The convolution operation \eqref{kernel conv} with $\mathbf{u}^{II}$ is considered, and the corresponding force is estimated by \eqref{kernel conv}
\vspace{-2mm}
\begin{equation}
    \bar{\mathbf{f}}^{II}\left(t_{n}^{\prime}\right) \approx K^{\prime} * \mathbf{u}_{\epsilon}^{II},
\vspace{-1mm}
\end{equation}
where $\bar{\mathbf{f}}^{II}$ is independent with stiff parameters $\epsilon$ in \eqref{stiff_ode}. Then, the macro step-size can be large and is not limited by $\mathbf{x}^{II}$ anymore. Note the accuracy of the remaining part $\mathbf{x}^{I}$ will be enhanced in the macro part which will be introduced momentarily.

\subsection{Proposed Algorithms for the Macro-Process}
\color{black}
This section introduces the detailed macro process for both fixed-step and variable-step simulations. 

\subsubsection{Fixed-step macro-process simulation}
After the kernel averaging in the micro-process, the value of the macro-state $\mathbf{u}$ at time $t_{n+1}$, i.e.,  $\mathbf{u}_{n+1}$, is well estimated by the Forward Euler scheme: 
\vspace{-2mm}
\begin{subequations}
\label{FE pre0}
\begin{align}
\label{FE pre}
    &\mathbf{u}^{II}_{n+1}=\mathbf{u}^{II}_{\epsilon}\left(t_{n}^{\prime}\right)+(H-\eta)\bar{\mathbf{f}}^{II}\left(t_{n}^{\prime}\right),\\ 
\label{FE pre2}
    &\mathbf{u}^{I}_{n+1}=\mathbf{u}^{I}_{\epsilon}\left(t_{n}^{\prime}\right)+(H-\eta)
          \mathbf{f}_{}^{I}\left(\mathbf{x}({t_{n}^{\prime})}, t_{n}^{\prime}\right),
\end{align}
\end{subequations}
where $\mathbf{u}_{\epsilon} =Q\mathbf{x}$ and the partial vector field $\mathbf{f}_{}^{I}\left(\mathbf{x}_{{n}^{\prime}}, t_{n}^{\prime}\right)$ can be obtained directly from the micro process simulation. While it depends on $\epsilon$ implicitly, it generally exhibits slow dynamics compared with $\mathbf{x}^{II}$. So using \eqref{FE pre2} without kernel averaging is acceptable, which also explains the reason for using identical operators for $\mathbf{x}^{I}$ previously. With $\mathbf{u}_{n+1}$ from \eqref{FE pre0} in hand, the accuracy of $\mathbf{f}_{}^{I}$ can be further enhanced by averaging two steps' vector fields by
\vspace{-2mm}
\begin{equation}
    \bar{\mathbf{f}}^{I}\left(t_{n}^{\prime}\right):=\frac{1}{2}\left[
    \mathbf{f}_{}^{I}\left(\mathbf{x}({t_{n}^{\prime})}, t_{n}^{\prime}\right)+
    \mathbf{f}_{}^{I}\left(R\mathbf{u}_{n+1}, t_{n+1}\right)\right].
\end{equation}
Then with $\bar{\mathbf{f}}=[\bar{\mathbf{f}}^{I},\bar{\mathbf{f}}^{II}]^\top$ in hand, the vector form of the macro process can be written as
\vspace{-2mm}
\begin{equation}
    \mathbf{u}_{n+1}=\mathbf{u}_{\epsilon}({t_{n}^{\prime})}+(H-\eta)\bar{\mathbf{f}}\left(t_{n}^{\prime}\right).
\vspace{-1mm}
\end{equation}
Compared with the step size $h_m$ in the micro process, the interval of a macro process $H$ is much larger which can accelerate the simulation speed without causing numerical stability problems.

\subsubsection{Variable-step macro-process simulation}

The macro-step can be controlled by an integral controller, enabling variable-step macro-process simulations.With \eqref{FE pre0} in hand, define the macro step size $Mh:= H-\eta$, $\mathbf{f}_{n}:=[ \mathbf{f}_{}^{I}\left(\mathbf{x}(t_{n}^{\prime}), t_{n}^{\prime}\right),~\bar{\mathbf{f}}^{II}\left(t_{n}^{\prime}\right)]^\top$, a zeroth-order error at $n^\text{th}$ step can be estimated by 
\begin{equation}
\label{error2}
r:=\left(\max \limits_{1\leq i \leq \mathbb{R}^{n_u}}\left|\frac{{Mh_n\mathbf{f}_{n}}}{|{[\mathbf{u}_\epsilon(t_{n}^{\prime})]}_i|+1}\right|\right) .
\end{equation}
Then, the local truncation error considering high-order errors and the control strategy of the macro-step can be estimated as follows, which can be used to control the truncation error induced by \eqref{FE pre0}
 \begin{align}
\label{s1}
\begin{split}
    e_n&=\frac{r}{1-r},\\
    \rho_{n+1}&=\min\left\{\left(\frac{Tol}{e_n}\right),\rho_{\text{max}}\right\},\\
    Mh_{n+1}&=\min\left\{\rho_{n+1} Mh_n,Mh_{\text{max}}\right\},
\end{split}
\end{align}
where $Mh_{\text{max}} =0.04$ and $\rho_{\text{max}}=1.05$ \cite{KaiyangPI}. Then the variables-step macro-process can be 
\begin{equation}
    \mathbf{u}_{n+1}=\mathbf{u}_{\epsilon}({t_{n}^{\prime})}+Mh_n\bar{\mathbf{f}}\left(t_{n}^{\prime}\right).
\vspace{-1mm}
\end{equation}
\color{black}
\color{black}
\subsection{Comparison with Other Multi-Timescale Methods}  

To accelerate EMT simulations, multi-rate EMT methods have been proposed and applied for decades \cite{Semlyen93,Mattavelli97,Qiuhua2018,Dewu2018,Yupeng2020,Huanfeng24}. To use these methods, the EMT model of a power system must be partitioned into a critical area that retains the detailed EMT model and an external area in which model reduction is required prior to simulations to generate a simplified model based on positive-sequence phasors\cite{Qiuhua2016,Qiuhua2018} or dynamic phasors \cite{Mattavelli97}. For instance, in phasor-EMT co-simulations, a phasor-based transient stability model is often adopted for the external area\cite{Qiuhua2016}. The critical and external areas are then co-simulated using either different simulators or settings, such as step sizes, via an interfacing algorithm. While these co-simulation methods are faster than a full EMT simulation of the entire system, they have notable limitations. First, the critical EMT area is typically determined offline to include, e.g., an inverter-based solar farm or wind farm for accurate simulations. Identifying the boundary of such a critical area can be challenging when IBRs highly penetrate and become dense in the power network\cite{Yupeng2020}. Second, interfacing two simulators across the boundary often introduces larger errors near the boundary \cite{Yupeng2020} or even leads to divergence\cite{Shengtao2012}. While increasing the number of iterations between the two simulators at each time step can reduce these errors, it also significantly increases the computational time. Third, to reduce the number of state variables, any model reduction algorithm must assume which dynamics are unimportant or ignorable. As a result, co-simulations may be unreliable or omit significant details if any assumptions are inaccurate or the model is not appropriately reduced. 

The proposed HMM framework takes a fundamentally different approach with the following advantages. First, the HMM approach does not require pre-defining model partition boundaries or reducing the model prior to simulation. Instead, it decomposes the simulation process in time by automatically alternating between a micro-process for detailed EMT simulation and a macro-process focusing on slower dynamics, both conducted on the same EMT model of the original system. Second, the HMM approach adaptively switches between the two processes at different timescales and adjusts the step size to ensure accuracy for both fast and slow dynamics. Theorem 1 provides theoretical guidance for selecting the lengths and step sizes of the micro-process and macro-process to achieve the desired performance. Third, the HMM framework can accommodate or complement other accelerated EMT simulation methods, such as the phasor-EMT co-simulation method. For example, while the external area is simulated using a phasor-based simulator, the critical EMT area can be further accelerated with the HMM approach.

\color{black}
\section{Case Studies}
In this section, the proposed HMM approach is first validated by simulations on a two-area system and the IEEE 39-bus system. The simulation results are benchmarked with simulations using the fourth-order Runge-Kutta (RK4) method with a small time step $h = 5~\mu\text{s}$. Then, the approach is tested on a large 390-bus system and compared with commercial tools including PSCAD and RK45.
 \color{black}All proposed algorithms are implemented in Matlab 2023b using scripts and tested on a desktop computer with $13^\text{th}$-generation Intel i7-13700K processor and 32GB RAM.\color{black}
\subsection{Test on Kundur's Two-area System}
The case study on this small system is to verify \eqref{Complexity} and \eqref{speedup} on the complexity and speedup using the proposed HMM approach. Replace generator G1 by the IBR model as shown in Fig.~\ref{twoarea}. Apply the HMM approach to simulate three contingencies (named S1, S2, and S3) which are three-phase bus grounding faults at buses 7, 8, and 9, lasting for 0.1 s and then being cleared after 5 cycles without losing any component. These contingencies are simulated for 10s.
\begin{figure}[!ht]
\centering
\includegraphics[width=0.7\columnwidth]{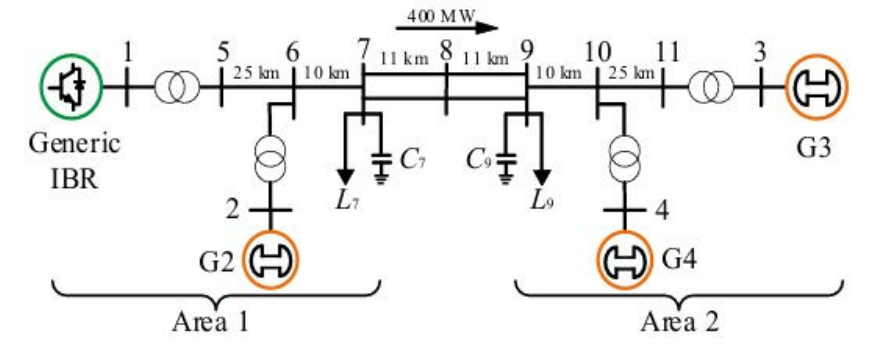}
\caption{Kundur's two-area system}
\label{twoarea}
\vspace{-3mm}
\end{figure}  
\begin{table}[!ht]
\begin{center}
\caption{Parameters for Kundur's Two-area System}
\label{p2area}
\begin{tabular}{c  c  c  c  c  c }
 \toprule
${h}$ & $L$ & $\eta$     & $C$ &  $D$ & $\theta_{\text {global}}$\\
\midrule
330 $\mu$s   & 30     & 0.0264 s & 3.0739 &  1.25       &  $\theta_1$        \\
\bottomrule
\end{tabular}
\end{center}
\end{table}

The parameters of the HMM approach are given in Table~\ref{p2area} for all contingencies. Selection of these parameters are explained as follows. Here $h$ is the time step for simulating the micro-state using  the DT-based the micro-solver, and is assumed constant to verify the complexity and speedup of the approach with different lengths of $H$. The selection of pair $(h,L)$ follows \cite{xiong2023semi} which demonstrates good accuracy and speed for the DT algorithm. The micro-interval $\eta$ should cover at least one period for a current or voltage waveform, i.e., 0.0167s under 60 Hz. Thus, $\eta$ is selected to cover 80 steps of $h$. $D$ is selected according to \cite{Ariel2008AMM}, and $C$ is uniquely determined to satisfy \eqref{intgral_condition}. Furthermore, the $\theta_{\text {global}} =\theta_1$ from generator G1 for simplicity. Although the proposed HMM approach is expected to provide significant acceleration in simulations, it may be also desirable to run a detailed EMT simulation for a brief initial time window to capture all fast transient dynamics under the fault. Thus, the HMM approach is activated at 1s afte the disturbance and applied for the rest of the simulation period.
\begin{figure}[!ht]
\centering
\subfloat[]{\includegraphics[width=0.48\columnwidth]{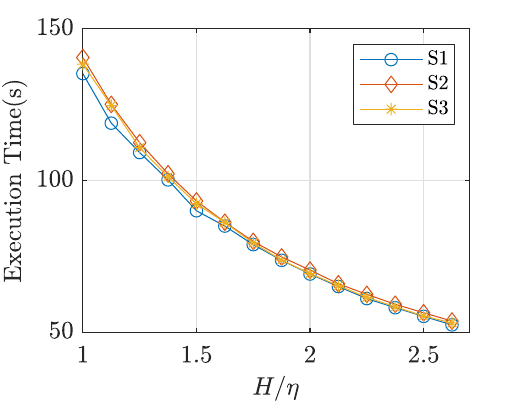}%
\label{fig:2area_time}}
\vspace{-5mm}
\subfloat[]{\includegraphics[width=0.48\columnwidth]{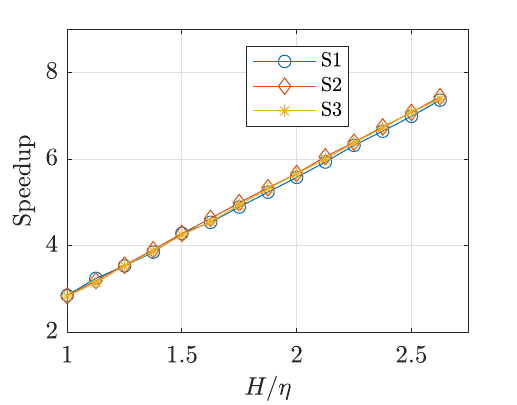}%
\label{fig:2area_speed}}
\hfil
\subfloat[]{\includegraphics[width=0.48\columnwidth]{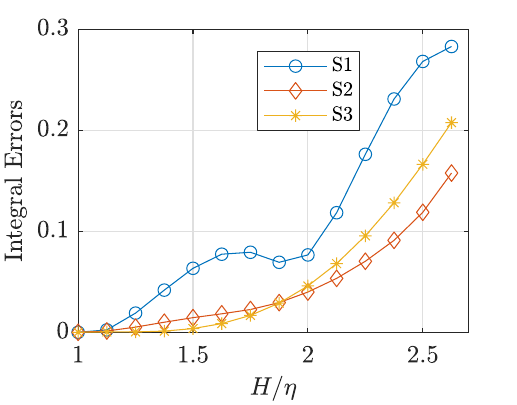}%
\label{fig:2area_error}}
\subfloat[]{\includegraphics[width=0.48\columnwidth]{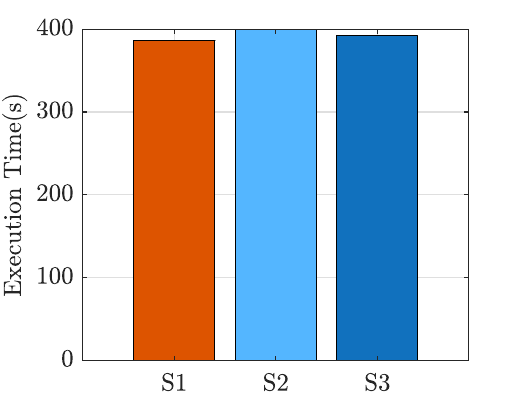}\label{fig:rktime}}%
\caption{Convergence analysis for HMM in two-area system under three scenarios. (a) Execution time under different $H$. (b) Speedup compared with the reference simulations. (c) Integral errors. (d) Execution time for reference simulations using a 10$\mu$s time step.}
\vspace{-1mm}
\label{fig:convergence_twoarea}
\end{figure}

To verify \eqref{Complexity} and \eqref{speedup}, different values of $H$ are tested, and the corresponding execution times and the speedup factor for three contingencies are shown in Fig. \ref{fig:2area_time} and Fig. \ref{fig:2area_speed}. The reference simulations use a 10$\mu$s time step as the benchmark of speedup of the proposed method. Fig.~\ref{fig:2area_error} confirms the convergence of the simulations by an error defined as the averaged integral of the infinity norm of the difference between the benchmark $\mathbf{x_{bh}}(t_n+\sum_{j=1}^{m-1}h_j)$ and $\mathbf{x_{m}}$ for all micro intervals along the simulation time $T$.

It can be concluded that for all contingencies, increasing $H$ can accelerate the simulation, and the speedup ratio is proportional to $H/\eta$, matching exactly the definition in \eqref{speedup}. Also, note that the complexity is proportional to the reciprocal of speedup when the micro step-size $h$ is fixed, so it is also proved well-defined as a dual form of the speedup.

In particular, simulation results on several state variables are shown in Fig.~\ref{fig:res_twoarea}, including rotor angle speeds and instantaneous line currents for S1, S2, and S3 with $H=2.625\eta$. From these figures, the micro-state and macro-state of the system are simulated alternatively with their respectively time steps $h$ and $H$. All simulated micro-state results match accurately the benchmark results. More importantly, for a majority of the simulation period, a much larger macro time step $H$, instead of $h$, is used without causing numerical instability. This leads to over eight-fold acceleration. Note that without using any reduced models created ahead of simulations, the proposed approach successfully reduces the EMT model \eqref{stiff_ode} on the fly of simulation  for significant speedup.

\begin{figure}[!ht]
\centering
\hspace{-5mm}
\subfloat[]{\includegraphics[width=0.36\columnwidth]{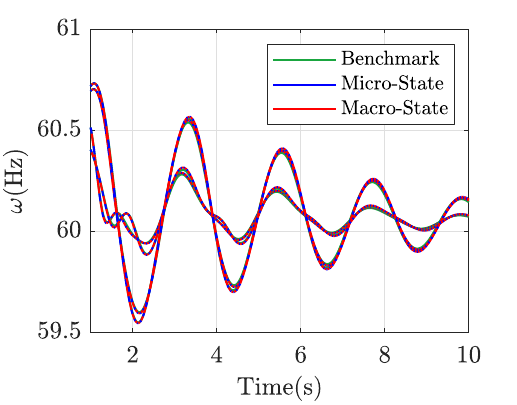}}%
\label{fig:S1Fregen}
\hspace{-5mm}
\subfloat[]{\includegraphics[width=0.36\columnwidth]{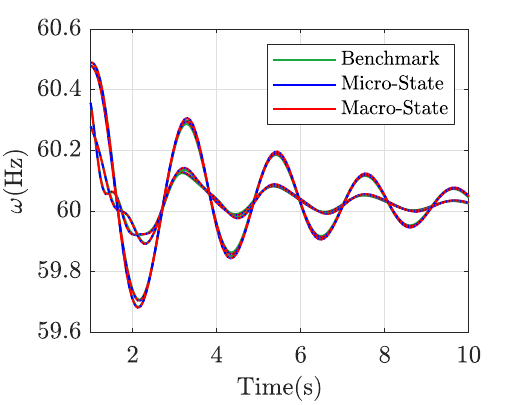}}%
\label{fig:S2Fregen}
\hspace{-5mm}
\subfloat[]{\includegraphics[width=0.36\columnwidth]{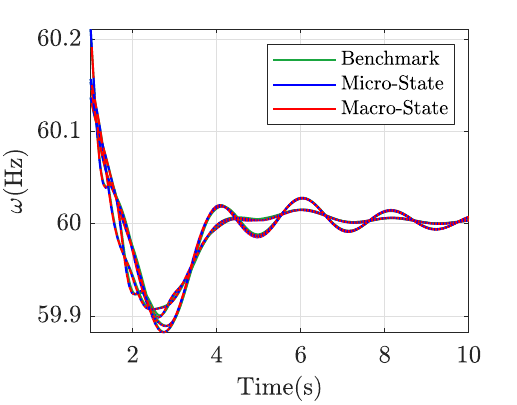}}%
\label{fig:S3Fregen}
\hspace{-5mm}

\hspace{-10mm}
\subfloat[]{\includegraphics[width=0.55\columnwidth]{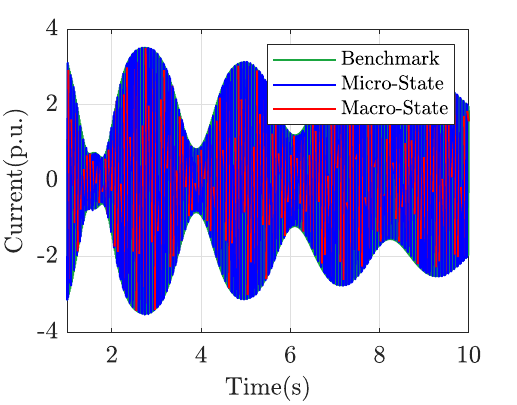}%
\label{fig:S1Current}}
\hspace{-5mm}
\subfloat[]{\includegraphics[width=0.55\columnwidth]{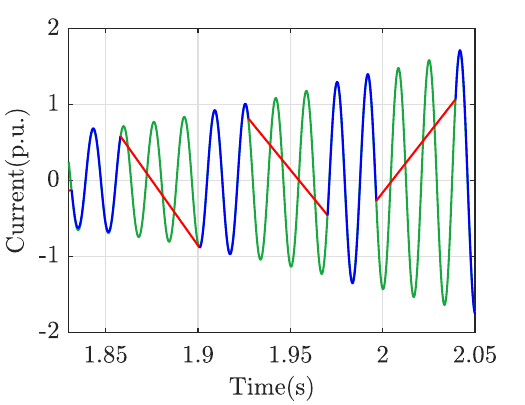}}%
\hspace{-10mm}

\hspace{-10mm}
\subfloat[]{\includegraphics[width=0.55\columnwidth]{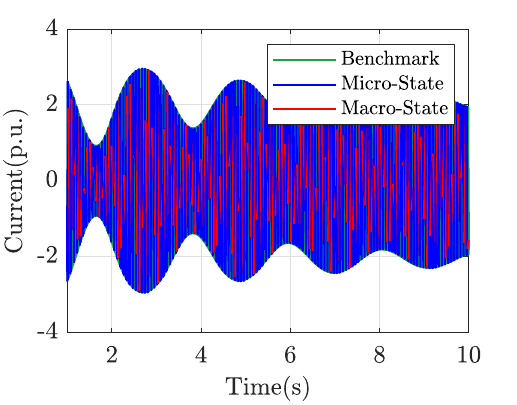}%
\label{fig:S2Current}}
\hspace{-5mm}
\subfloat[]{\includegraphics[width=0.55\columnwidth]{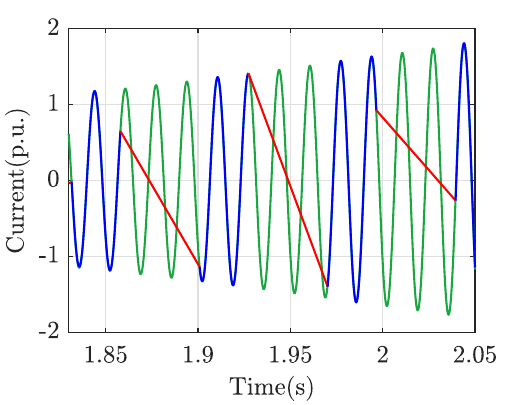}}%
\hspace{-10mm}

\hspace{-10mm}
\subfloat[]{\includegraphics[width=0.55\columnwidth]{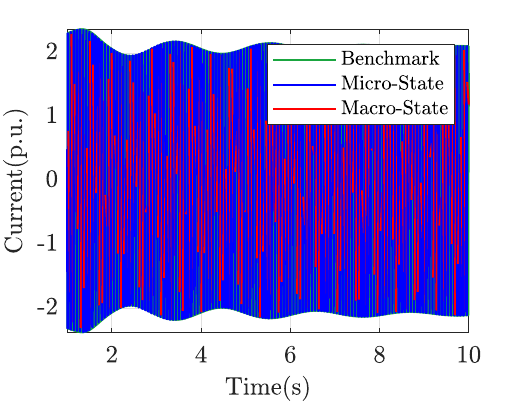}%
\label{fig:S3Current}}
\hspace{-5mm}
\subfloat[]{\includegraphics[width=0.55\columnwidth]{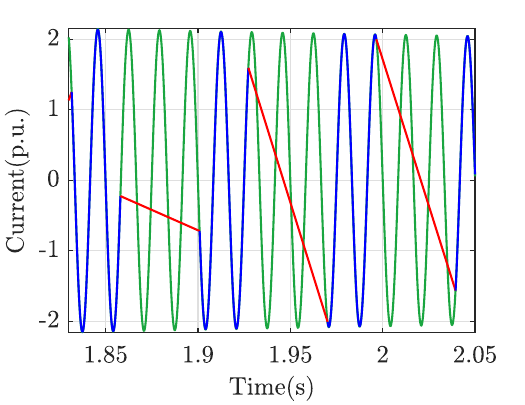}}%
\hspace{-10mm}

\caption{Simulation results for the two-area system under S1, S2, and S3. (a) Rotor angular frequency in S1. (b) Rotor angular frequency in S2. (c) Rotor angular frequency in S3. (d) Phase A current of the branch 7-8 for S1. (e) Zoomed-in view of (d). (f) Phase A current of the branch 7-8 for S2. (g) Zoomed-in view of (f). (h) Phase A current of the branch 7-8 for S3. (i) Zoomed-in view of (h).}
\vspace{-4mm}
\label{fig:res_twoarea}
\end{figure}

\subsection{Test on the IEEE 39-bus System}
\begin{figure}[!ht]
\centering
\includegraphics[width=0.65\columnwidth]{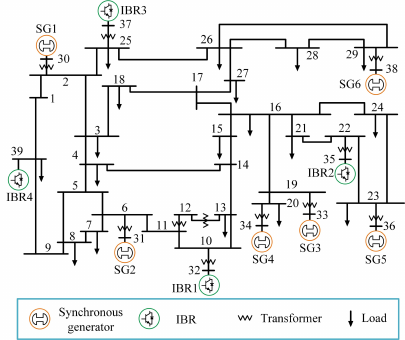}
\caption{A modified IEEE 39-bus System}
\label{10_39diag}
\vspace{-3mm}
\end{figure}  
The scalability of the proposed HMM approach is tested on a modified IEEE 39-bus system that has its penetration level of renewable energy reaching 50\% by replacing four generators by IBRs as shown in Fig.~\ref{10_39diag}. Simulate these two contingencies:
\begin{itemize}
  \item S4: At $t=$0.1 s, a three-phase grounding fault is applied at bus 16 and then cleared after 5 cycles; at $t=$8 s, the load at bus 3 is disconnected and then reconnected after 5 cycles; the simulation ends at $t=$16 s.
  \item S5: At $t=$0.1 s, generator SG4 at bus 34 is tripped permanently; the simulation ends at $t=$8 s.
\end{itemize}
Use the same set of parameters in Table~\ref{p2area}. Here, the proposed defect control in micro-process is considered with a tolerance $\epsilon_1 =10^{-2}$. The variable-step macro-process is enabled in this test. By varying $Tol$, the resulting speedup is given in Table~\ref{table:HMM_39} compared to a full EMT simulation using a $10 \mu s$ step size RK method. \color{black}Note that the micro-state from the HMM simulation results highly match the benchmark results, as observed from the small errors. Here the maximum error among all state variables is defined as the infinity norm $\left(\|\cdot\|_{\infty}\right)$ of the difference between the results simulated by the HMM approach and the benchmark. This metric provides a strict measure of convergence and highlights the accuracy of the proposed approach across all state variables at every simulation step. \color{black}For both S4 and S5, a large $Tol$ allows a larger average macro time step and thus achieves a faster speed. The simulation results of S4 with $Tol=10^{-2}$ and S5 with $Tol=10^{-2}$ are shown in Fig.~\ref{fig:39_S4} and Fig.~\ref{fig:39_S5}, respectively. \color{black}
For both contingencies, it is observed that the micro-dynamics become dominant when a fault occurs and immediately after it is cleared, and a large $H$ adaptively rising to 0.04 s is achieved without causing any numerical instability seen from Fig.~\ref{fig:S4step} and Fig.~\ref{fig:S5step}. It demonstrates that the proposed approach adaptively adjusts the macro-step size based on real-time system dynamics, ensuring an appropriate resolution of simulation while maintaining computational efficiency. Thus, it is recommended that a micro-state simulation, i.e., a full EMT simulation, is performed from the beginning of a fault until it is cleared in order to capture all fast EMT dynamics accurately under the fault. Then, when fast dynamics decay and slower dynamics become dominant, the HMM simulation may start to switch between the micro-state simulation and macro-state simulation.
\color{black}
The simulation is directly performed on the full EMT model without any explicit model reduction. The time step 0.04 s is already in the range of time steps for transient stability simulations while EMT dynamics are still available with the micro-state results during the intervals of $\eta$. 

\begin{table}[!ht]
\vspace{-2mm}
\begin{center}
\setlength{\tabcolsep}{2pt}
\caption{Computational Performance of the HMM Method}

\label{table:HMM_39}
     \begin{tabular}{c c c c c c}
        \toprule
\thead {Case }  
    & {\thead{$Tol$} }
        & {\thead{Execution\\Time(s) }}  & {\thead{Avg. Macro\\ Step-size }} & {\thead{Integral\\Error}}  & {\thead{Speedup}} \\ [0.5ex]
        \midrule
S4      & $10^{-2}$ & 91.9944 & 0.0248&$3.4171\times10^{-3}$&6.3766   \\
        & $7.5\times10^{-3}$ & 93.7706 & 0.0235& $1.5965\times10^{-3}$&6.2558 \\
        & $2.5\times10^{-3}$ & 110.5661 & 0.0149&$9.9184\times10^{-4}$&5.3055\\
S5      & $1\times10^{-1}$   &  46.2250 & 0.0257& $7.8574\times10^{-3}$ &6.5621  \\
        & $1\times10^{-2}$ &  49.8745 & 0.0203& $5.0552\times10^{-4}$ &6.0818\\
        & $1\times10^{-3}$ &  76.6401 & 0.0028&$5.5834\times10^{-5}$& 3.9579\\
        \bottomrule
    \end{tabular}
    \end{center}
    \vspace{-2mm}
\end{table}

\begin{figure}[!ht]
\centering
\hspace{-10mm}
\subfloat[]{\includegraphics[width=0.55\columnwidth]{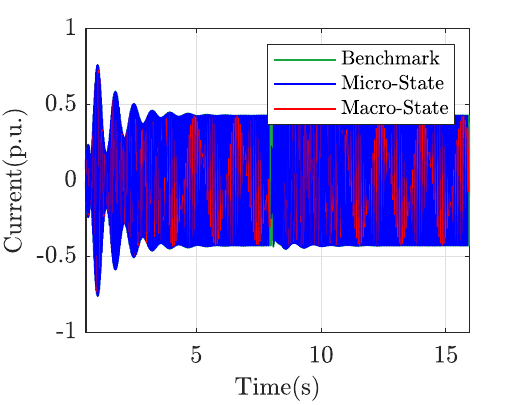}%
\label{fig:S4Current}}
\hspace{-5mm}
\subfloat[]{\includegraphics[width=0.55\columnwidth]{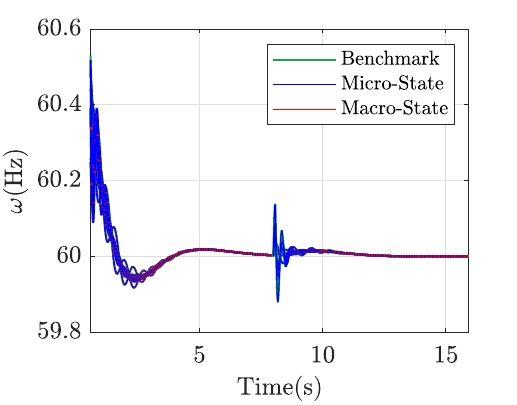}}%
\hspace{-10mm}

\color{black}
\hspace{-10mm}
\subfloat[]{\includegraphics[width=0.55\columnwidth]{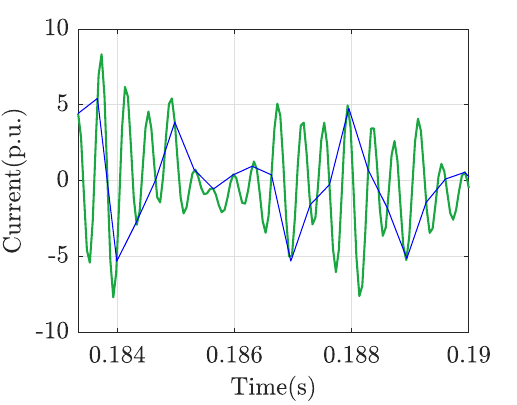}%
\label{fig:S4Current_zoomin3}}
\hspace{-5mm}
\subfloat[]{\includegraphics[width=0.55\columnwidth]{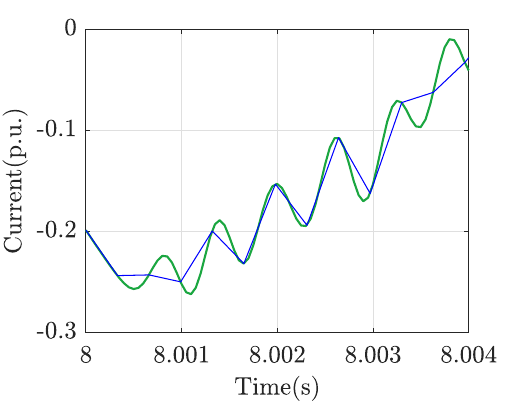}}%
\hspace{-10mm}

\hspace{-10mm}
\subfloat[]{\includegraphics[width=0.55\columnwidth]{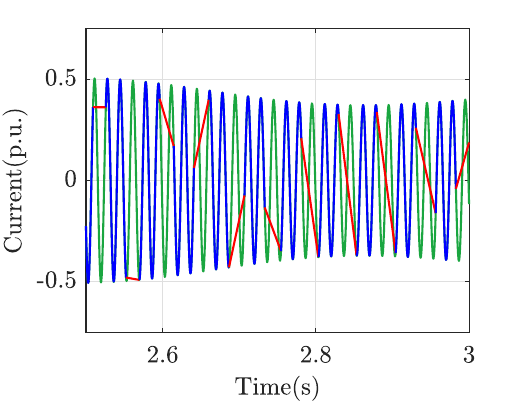}%
\label{fig:S4Current_zoomin1}}
\hspace{-5mm}
\subfloat[]{\includegraphics[width=0.55\columnwidth]{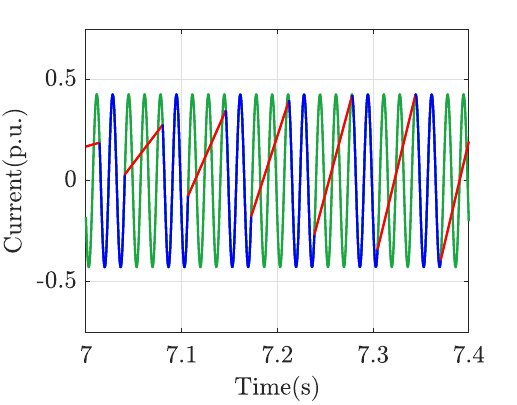}}%
\hspace{-10mm}

\hspace{-10mm}
\subfloat[]{\includegraphics[width=0.55\columnwidth]{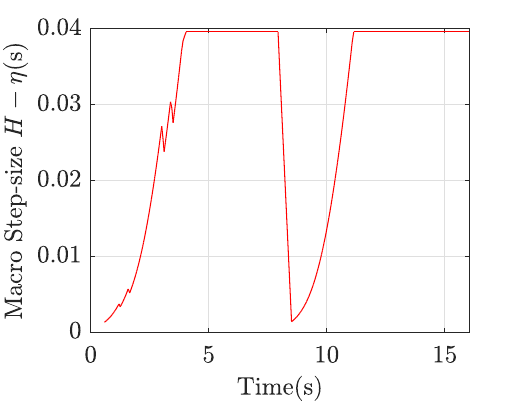}%
\label{fig:S4step}}
\hspace{-5mm}
\subfloat[]{\includegraphics[width=0.55\columnwidth]{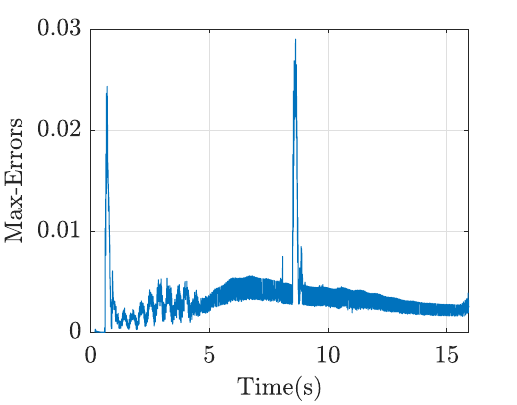}}%
\hspace{-10mm}

\caption{Simulation results for S4. (a) Phase A current of the branch line 16-24. (b) Rotor angular frequencies. (c) Zoomed-in view of (a) near the fault. (d) Zoomed-in view of (a) near the generator tripping. (e) Zoomed-in view of (a) under oscillation. (e) Zoomed-in view of (a) under steady state. (g) Macro step size during the simulation. (h) Maximum errors among all state variables. }
\label{fig:39_S4}
\color{black}
\end{figure}

\begin{figure}[!ht]
\centering
\hspace{-10mm}
\subfloat[]{\includegraphics[width=0.55\columnwidth]{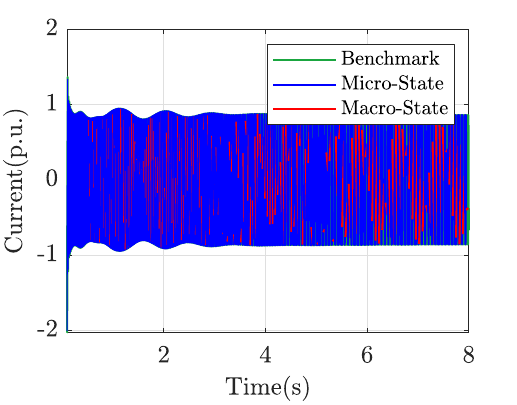}%
\label{fig:S5Current}}
\hspace{-5mm}
\subfloat[]{\includegraphics[width=0.55\columnwidth]{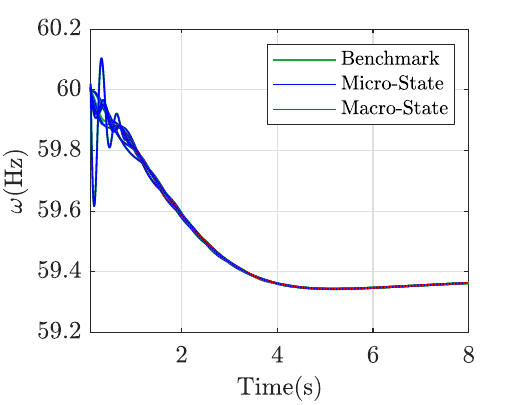}}%
\hspace{-10mm}

\color{black}
\hspace{-10mm}
\subfloat[]{\includegraphics[width=1.1\columnwidth]{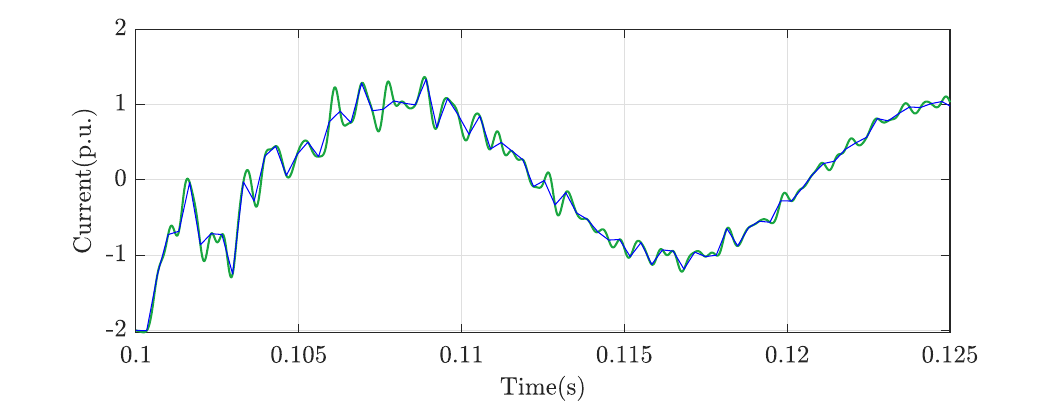}%
\label{fig:S5Current_zoomin3}}
\hspace{-10mm}

\hspace{-10mm}
\subfloat[]{\includegraphics[width=0.55\columnwidth]{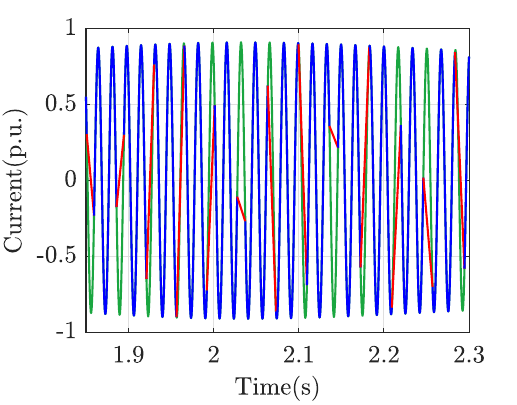}%
\label{fig:S5Current_zoomin1}}
\hspace{-5mm}
\subfloat[]{\includegraphics[width=0.55\columnwidth]{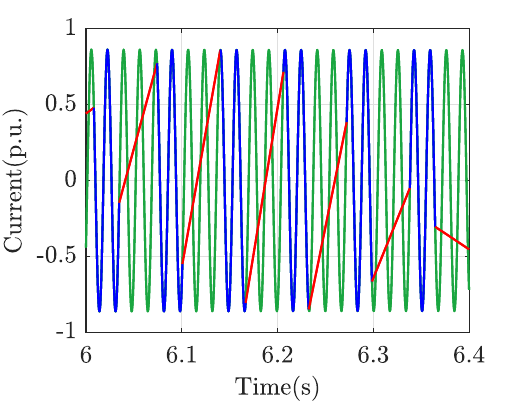}}%
\hspace{-10mm}

\hspace{-10mm}
\subfloat[]{\includegraphics[width=0.55\columnwidth]{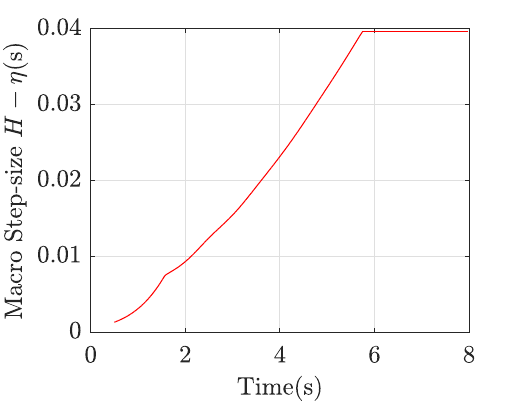}%
\label{fig:S5step}}
\hspace{-5mm}
\subfloat[]{\includegraphics[width=0.55\columnwidth]{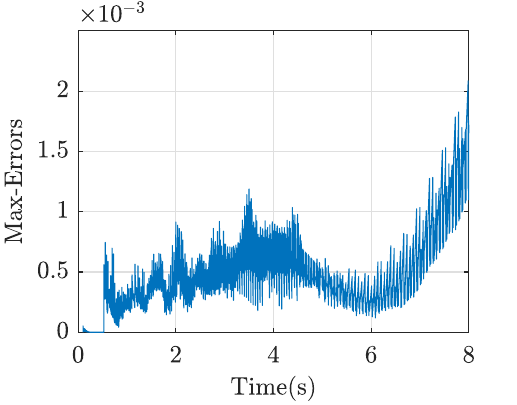}}%
\hspace{-10mm}

\caption{Simulation results for S5. (a) Phase A current of the branch line 17-18. (b) Rotor angular frequencies. (c) Zoomed-in views of (a) near the fault. (d) Zoomed-in view of (a) under oscillation. (e) Zoomed-in view of (a) under steady-state. (f) Macro step size during the simulation. (g) Maximum errors among all state variables. }

\label{fig:39_S5}
\end{figure}
\color{black}
\subsection{Test on the Large-scale System}
The performance of the HMM approach on the large-scale systems is further validated on a synthetic 390-bus system \cite{xiong2023semi}. To compare with commercial software such as PSCAD, the proposed algorithm is translated from its MATLAB code into C code, and then compiled as an executable MEX file. The PSCAD codes used in this section are complied with GFortran 8.1 (64-bit) in PSCAD 5.0.2. A new contingency is simulated, named S6: At $t=0.1$s a three-phase grounding fault is applied at bus 10 and then cleared after 0.2s, a 1 Hz forced oscillation happens due to the external three-phase current source injection $0.32\sin(2\pi t)$ [p.u.] (Phase A) at bus 30, and the simulation ends at $t=8$s. The HMM approach is compared with the nodal formulation-based approach \cite{Dommel1969} utilized in PSCAD and the RK45 method \cite{dormand1980family} utilized in Simulink. The simulation results are benchmarked with simulations using the nodal formulation-based approach utilized in PSCAD with a small step size $h = 1~\mu\text{s}$. The time costs for the simulation and maximum normalized errors among all state variables are presented in Table~\ref{table:errorcomp2}. Also, the phase-A voltage results at the grounding bus are shown in Fig~\ref{fig:390_S6_1} and Fig~\ref{fig:390_S6_2}. The results in Table~\ref{table:errorcomp2} highlight the advantages of the proposed HMM approach in terms of speed and accuracy, for a comparable error tolerance i.e., \text{Max Error} = 8.29\%, the HMM approach takes a CPU time of 176.59 s, whereas PSCAD requires 12725.77s using a 5 $\mu$s step size. This demonstrates that the HMM approach is approximately 70 times faster than PSCAD while achieving a similar accuracy level. Moreover, the RK45 diverges for a relaxed tolerance of \( \text{tol} = 10^{-1} \), indicating numerical instability in the stiff simulation. For a tighter tolerance (\( \text{tol} = 10^{-3} \)), RK45 has a maximum error of 43.60\%\ but at a much higher computational cost of 5246.13 s. In comparison, the HMM approach has a significantly lower error and an almost 30 times speedup compared to RK45. Note that PSCAD's time performance improves as the step size increases. However, this comes at the cost of accuracy, as larger step sizes (\(h = 25\mu \)s to \(h = 100\mu \)s) result in progressively larger errors. This can be observed from transients shown in Fig.~\ref{fig:390_S6_1}, only PSCAD with \(h = 5\mu\)s can match the benchmark in a good shape. The proposed HMM approach, however, utilizes a high-order approximation in the micro-process and provides accurate results for fast dynamics under large step sizes. Note that, if needed, detailed high-resolution micro-state dynamics can be reconstructed from the macro-process via the micro-solver. Also, the variable-step algorithm ensures that the step size of the macro-process is dynamically adjusted to maintain simulation accuracy. At the moment the fault happens, the step size of the macro-process becomes smaller to accommodate the need for higher resolution and accuracy, which is also shown in Fig.~\ref{fig:S6step}.  Then, the macro-step size gradually increases up to 0.04 s which can accurately simulate slow dynamics of the forced oscillation as shown in Fig.\ref{fig:S6Current} and Fig.~\ref{fig:S6Current_zoomin3}. The HMM approach achieves a balanced trade-off between computational speed and accuracy, outperforming traditional fixed-step methods, such as the nodal formulation method in PSCAD, and variable-step methods, such as the RK45 method in Simulink. The results demonstrate the superiority of the proposed HMM approach for accelerated, accurate two-timescale power system simulations including both fast dynamics and slow dynamics.



\color{black}
\begin{table}[!ht]
\vspace{-2mm}
\begin{center}
\setlength{\tabcolsep}{0pt}
\color{black}
\caption{Computational performance by Different Methods for 390-bus System}
\label{table:errorcomp2}
\begin{tabular}{lccc}
        \toprule
 Approach & CPU Time$(\mathrm{s})$ & Max Error(\%)\\
\midrule
HMM & 176.59  & 0.0829 \\
RK45 $(tol=10^{-1})$ & Diverged & Diverged\\
RK45 $(tol=10^{-3})$ &  5246.133 & 0.4360 \\
PSCAD $(h=5\mu s)$ & 12725.77 & 0.0884 \\
PSCAD $(h=25\mu s)$ & 2701.64 & 0.7832 \\
PSCAD $(h=50\mu s)$ & 1354.55 & 1.4712\\
PSCAD $(h=75\mu s)$ & 955.39 & 1.6350 \\
PSCAD $(h=100\mu s)$ & 649.97 & 1.4688\\
\bottomrule
\end{tabular}
    \end{center}
\end{table}

\vspace{-5mm}
\color{black}



\begin{figure}[!ht]
\centering

\color{black}
\hspace{-10mm}
\subfloat[]{\includegraphics[width=1.1\columnwidth]{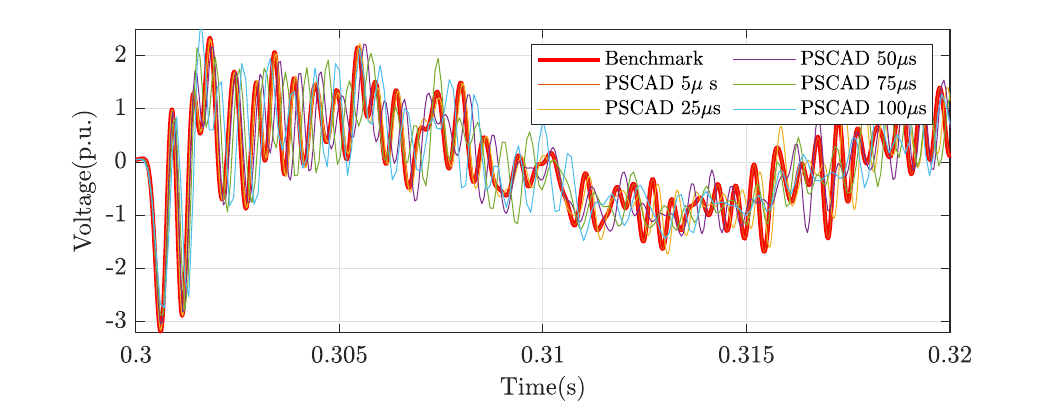}%
\label{fig:S6voltage_zoomin1}}
\hspace{-10mm}

\color{black}
\hspace{-10mm}
\subfloat[]{\includegraphics[width=1.1\columnwidth]{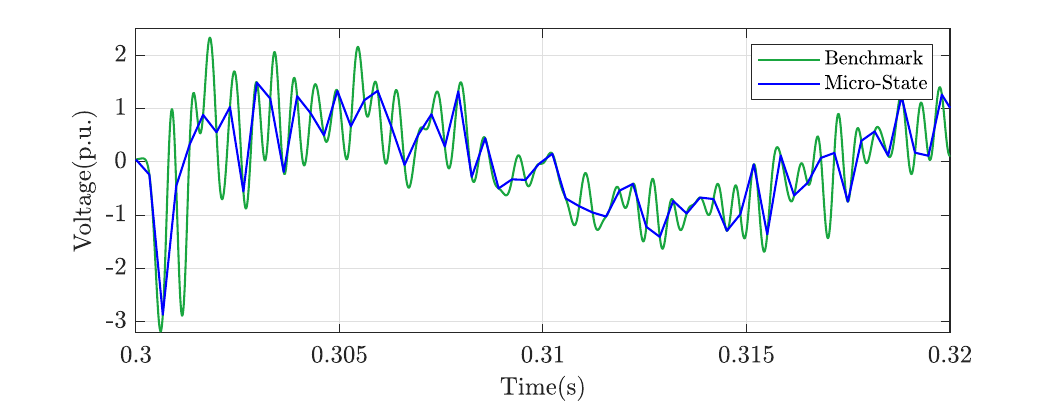}%
\label{fig:S6voltage_zoomin2}}
\hspace{-10mm}
\caption{Simulation results of Phase-A voltage at the grounding Bus for S6 after the fault. (a) results simulated by PSCAD using different step sizes. (b) results simulated by HMM. }
\vspace{-5mm}
\label{fig:390_S6_1}
\end{figure}

\begin{figure}[!ht]
\centering

\subfloat[]{\includegraphics[width=0.5\columnwidth]{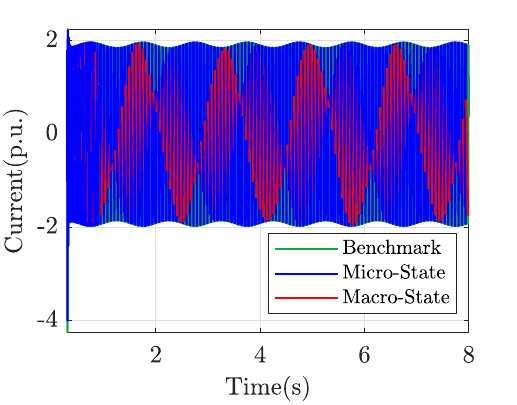}%
\label{fig:S6Current}}
\hspace{-5mm}
\subfloat[]{\includegraphics[width=0.5\columnwidth]{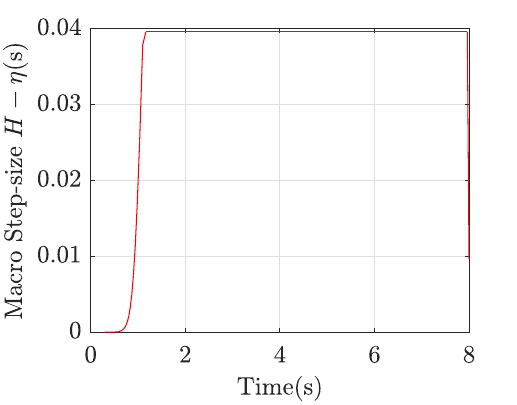}\label{fig:S6step}}%

\color{black}
\hspace{-10mm}
\subfloat[]{\includegraphics[width=1.1\columnwidth]{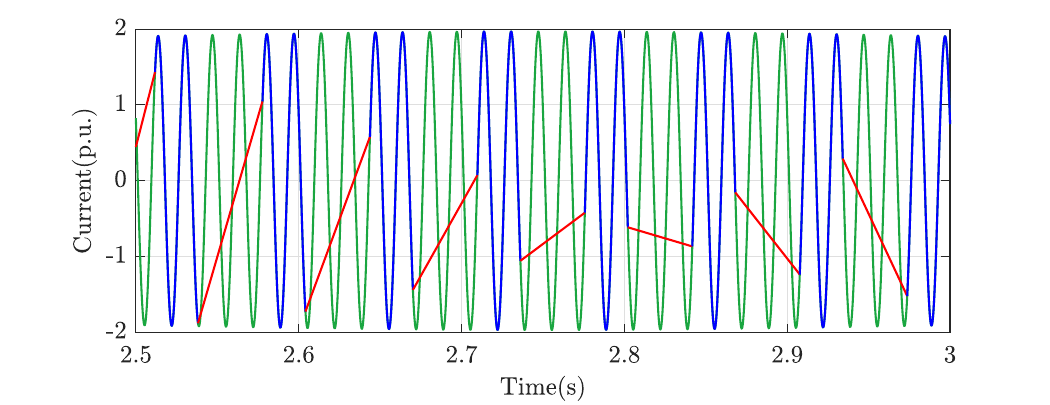}%
\label{fig:S6Current_zoomin3}}
\hspace{-10mm}
\caption{Simulation results for S6. (a) Phase A current of the branch line 1-2. (b) Macro step size during the simulation. (c) Zoomed-in views of (a). (d) Zoomed-in view of (a) near the steady state. }
\vspace{-5mm}
\label{fig:390_S6_2}
\end{figure}

\section{Conclusion}
This paper has proposed an HMM approach for efficient multi-timescale simulation of power systems with penetration of IBRs. On the EMT model of a power system, the HMM approach takes significantly enlarged time steps to evolve its macroscopic behavior for most of the simulation period when microscopic dynamics are either monotonic or predictable following a contingency. By leveraging a DT-based variable-step semi-analytical solver, the proposed HMM framework links detailed-level micro-state simulations over disconnected intervals via macro-state solutions to ensure numerical stability throughout the simulation, and it enables automatic zooming in and out between macroscopic and microscopic dynamics on the fly during the simulation. 
\color{black}The proposed HHM framework is applicable to any dynamical system with an available ODE model. Thus, the proposed HHM approach may be extended and enhanced to simulate additional timescales in future power systems with integrated IBRs, addressing more detailed simulations of inverter switching cycles. \color{black}
Case studies have demonstrated that the proposed HMM approach can be directly applied to the full EMT model of a power system and can more efficiently capture accurate fast dynamics while skipping over less important EMT details to speed up the simulation. 



\bibliographystyle{IEEEtran}



\begin{IEEEbiography}[{\includegraphics[width=1in,height=1.25in,clip,keepaspectratio]{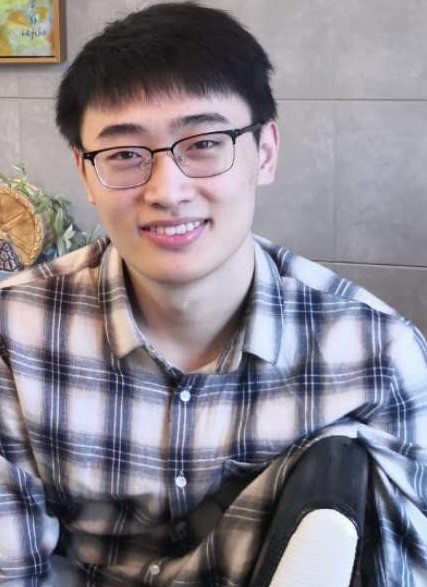}}]{Kaiyang Huang}
(Graduate Student Member, IEEE) received the
B.S. degree in electrical engineering
from North China Electric Power University, China, in 2020.
He is currently pursuing the Ph.D. degree at the
Department of Electrical Engineering and Computer
Science, University of Tennessee, Knoxville, USA.
His research interests include power system simulation, transient
stability analysis, and dynamics.  \end{IEEEbiography}

\begin{IEEEbiography}[{\includegraphics[width=1in,height=1.25in,clip,keepaspectratio]{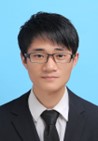}}]{Min Xiong}
(Member, IEEE) received the B.S. and M.S. degrees from Wuhan University, Wuhan, China, in 2013 and 2016, respectively, and the Ph.D. degree from the University of Tennessee, Knoxville, TN, USA, in 2023, all in electrical engineering. He was an engineer with State Grid Hubei Power Company from 2016 to 2019. He is now a postdoctoral researcher at the Power Systems Engineering Center, National Renewable Energy Laboratory, Golden, CO, USA. His research interests include electrical parameter measurement, relay protection, power system stability analysis, electromagnetic transient simulation, and integration of renewable resources.  \end{IEEEbiography}

\begin{IEEEbiography}[{\includegraphics[width=1in,height=1.25in,clip,keepaspectratio]{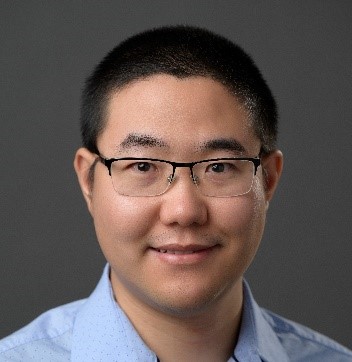}}]{Yang Liu}
(Member, IEEE)  is currently a principal engineer at Quanta Technology LLC. He received his B.S. degree in energy and power engineering from Xi’an Jiaotong University, China, in 2013, his M.S. degree in power engineering from Tsinghua University, China, in 2016, and his Ph.D. degree in electrical engineering from the University of Tennessee, Knoxville, USA, in 2021. He was a Post-Doctoral Researcher at the University of Tennessee in 2022 and at Argonne National Laboratory in 2023. His research interests include power system simulation, dynamics, stability, and control. \end{IEEEbiography}

\begin{IEEEbiography}[{\includegraphics[width=1in,height=1.25in,clip,keepaspectratio]{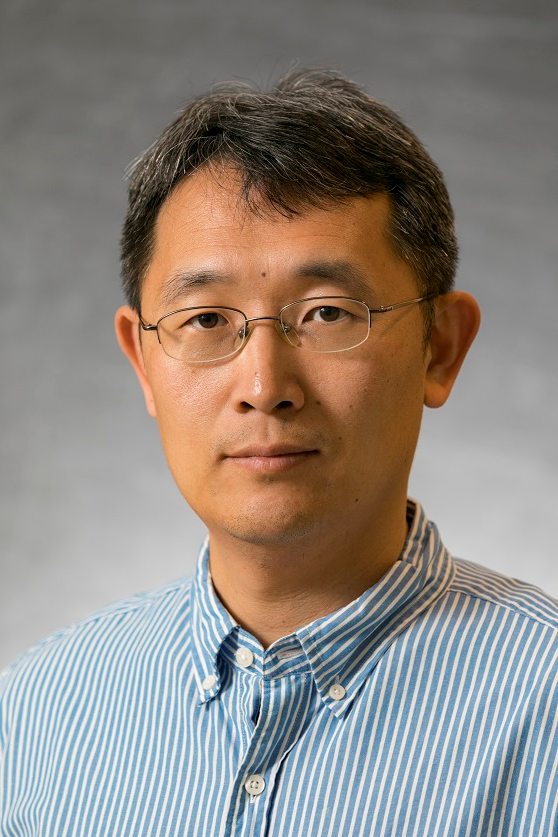}}]{Kai Sun}
(Fellow, IEEE) is a professor in the Department of Electrical Engineering and Computer Science at the University of Tennessee, Knoxville, USA. He received his B.S. degree in Automation in 1999 and his Ph.D. degree in Control Science and Engineering in 2004, both from Tsinghua University, Beijing, China. From 2007 to 2012, he served as a Project Manager in grid operations, planning, and renewable integration at the Electric Power Research Institute (EPRI), Palo Alto, CA. \end{IEEEbiography}

\vfill

\end{document}